\documentclass[a4paper,accepted=2026-06-15,onecolumn,10pt,groupedaddress,nofootinbib,
allowfontchangeintitle,
unpublished
]{quantumarticle}
\pdfoutput=1

\usepackage[english]{babel}

\usepackage{amsmath,amsfonts,amssymb,amsthm,bbm}
\usepackage{mathtools}
\usepackage{mathrsfs}
\usepackage[normalem]{ulem}
\usepackage{enumerate}  
\usepackage[pdftex]{graphicx} 
\usepackage{rotating}  
\usepackage[dvipsnames]{xcolor}
\usepackage{relsize}
\usepackage{comment}
\usepackage{tikz}
\input{tikzit.sty}

\tikzstyle{medium rectangle}=[fill=white, draw=black, shape=rectangle, minimum width=0.75 cm, minimum height=1 cm]
\tikzstyle{red}=[fill=red, draw=black, shape=circle]
\tikzstyle{system_label}=[fill=none, draw=none, shape=circle]
\tikzstyle{small square}=[fill=white, draw=black, shape=rectangle]
\tikzstyle{big rectangle}=[fill=white, draw=black, shape=rectangle, minimum width=1.2 cm, minimum height=1.5 cm]
\tikzstyle{tall rectangle}=[fill=white, draw=black, shape=rectangle, minimum width=1.2 cm, minimum height=2.1 cm]
\tikzstyle{medium square}=[fill=white, draw=black, shape=rectangle, minimum width=0.6 cm, minimum height=0.6 cm]
\tikzstyle{huge rectangle}=[fill=white, draw=black, shape=rectangle, minimum width=2 cm, minimum height=4.5 cm]
\tikzstyle{yellow small}=[fill={rgb,255: red,255; green,252; blue,144}, draw=black, shape=rectangle]
\tikzstyle{violet node}=[fill={rgb,255: red,195; green,187; blue,255}, draw=black, shape=rectangle, minimum width=1.2 cm, minimum height=1.5 cm]
\tikzstyle{pink node}=[fill={rgb,255: red,238; green,174; blue,255}, draw=black, shape=rectangle, minimum width=1.2 cm, minimum height=1.5 cm]
\tikzstyle{pompelmo node}=[fill={rgb,255: red,255; green,187; blue,166}, draw=black, shape=rectangle, minimum width=1.2 cm, minimum height=1.5 cm]
\tikzstyle{ottano node}=[fill={rgb,255: red,0; green,195; blue,195}, draw=black, shape=rectangle, minimum width=0.8 cm, minimum height=0.5 cm]
\tikzstyle{blue}=[fill={rgb,255: red,214; green,201; blue,255}, draw=black, shape=rectangle]
\tikzstyle{vert-rectangle}=[fill=white, draw=black, shape=rectangle, minimum width=1.8 cm, minimum height=0.8 cm]
\tikzstyle{ott-small}=[fill={rgb,255: red,0; green,195; blue,195}, draw=black, shape=rectangle]
\tikzstyle{ott-wide}=[fill={rgb,255: red,0; green,195; blue,195}, draw=black, shape=rectangle, minimum width=2.8 cm]

\tikzstyle{lambda}=[-, draw={rgb,255: red,191; green,191; blue,191}]
\tikzstyle{state}=[<-]
\tikzstyle{bluette}=[-, fill={rgb,255: red,214; green,201; blue,255}, draw={rgb,255: red,192; green,181; blue,229}]
\tikzstyle{greenish}=[-, fill={rgb,255: red,160; green,217; blue,255}, draw={rgb,255: red,137; green,188; blue,219}]
\tikzstyle{white}=[-, fill=white]
\tikzstyle{reddish}=[-, fill={rgb,255: red,255; green,143; blue,145}, dashed]
\tikzstyle{yellowish}=[-, fill={rgb,255: red,255; green,252; blue,144}, dashed, line width=1 pt]
\tikzstyle{red non-dashed}=[-, fill={rgb,255: red,255; green,143; blue,145}]
\tikzstyle{violet}=[-, dashed, fill={rgb,255: red,195; green,187; blue,255}]
\tikzstyle{green non-dashed}=[-, fill={rgb,255: red,160; green,217; blue,255}]
\tikzstyle{ontological}=[-, fill={Red!20}, draw={Red!40}]
\tikzstyle{pink}=[-, fill={rgb,255: red,238; green,174; blue,255}, dashed]
\tikzstyle{blue non-dashed}=[-, fill={rgb,255: red,178; green,255; blue,246}]
\tikzstyle{yell non-dashed}=[-, fill={rgb,255: red,255; green,248; blue,137}]
\tikzstyle{ottano}=[-, fill={rgb,255: red,164; green,184; blue,255}, draw={rgb,255: red,147; green,168; blue,229}]
\tikzstyle{pompelmo}=[-, fill={rgb,255: red,255; green,187; blue,166}, draw={rgb,255: red,213; green,156; blue,139}]
\tikzstyle{dark blue}=[-, fill={rgb,255: red,0; green,195; blue,195}, draw={rgb,255: red,0; green,98; blue,98}]
\tikzstyle{d.blue non-dashed}=[-, fill={rgb,255: red,237; green,148; blue,112}]
\tikzstyle{dashed edge}=[-, loosely dashed]
\tikzstyle{ottano non-drawn}=[-, fill={rgb,255: red,164; green,184; blue,255}]
\tikzstyle{norm_dotted}=[-, dotted, thick]
\tikzstyle{grey_densely_dashed}=[-, dashed, draw={rgb,255: red,156; green,156; blue,156}]
\tikzstyle{dot_grey}=[-, loosely dotted, draw={rgb,255: red,156; green,156; blue,156}, thick]
\tikzstyle{q-wire}=[-, draw={rgb,255: red,79; green,4; blue,134}, line width=1 pt]
\tikzstyle{m-wire}=[-, draw={rgb,255: red,139; green,185; blue,0}, line width=1.5 pt]
\tikzstyle{ont-wire}=[-, draw={black!40}, line width=0.9 pt]

\usepackage{breakurl}

\usepackage[numbers,sort&compress]{natbib}
\usepackage[colorlinks=true]{hyperref} 
\usepackage{orcidlink}

\usepackage{dsfont}
\usepackage{rotating}
\usepackage{floatpag}
\rotfloatpagestyle{empty}

\usepackage{hyperref}

\usepackage{amsmath,amsthm,amssymb}
\usepackage{xspace,enumerate,color,epsfig}
\usepackage{graphicx}
\usepackage{marvosym}

\usepackage{tikzfig}
\usepackage{docmute}
\usepackage{keycommand}

\usepackage{pifont}

\usepackage{enumitem}

\newkeycommand{\pointmap}[style=point][1]{\,\tikz{\node[style=\commandkey{style}] (x) at (0,-0.3) {$#1$};\node [style=none] (3) at (0, 1.0) {};\node [style=none] (3) at (0, -1.0) {};\draw (x)--(0,0.7);}\,}

\newkeycommand{\typedkpoint}[style=kpoint,edgestyle=][2]{%
\,\begin{tikzpicture}
  \begin{pgfonlayer}{nodelayer}
    \node [style=none] (0) at (0, 1) {};
    \node [style=\commandkey{style}] (1) at (0, -0.75) {$#2$};
    \node [style=right label] (2) at (0.25, 0.5) {$#1$};
  \end{pgfonlayer}
  \begin{pgfonlayer}{edgelayer}
    \draw [\commandkey{edgestyle}] (1) to (0.center);
  \end{pgfonlayer}
\end{tikzpicture}\,}

\newkeycommand{\typedkpointdag}[style=kpoint adjoint,edgestyle=][2]{%
\,\begin{tikzpicture}
  \begin{pgfonlayer}{nodelayer}
    \node [style=none] (0) at (0, -1) {};
    \node [style=\commandkey{style}] (1) at (0, 0.75) {$#2$};
    \node [style=right label] (2) at (0.25, -0.5) {$#1$};
  \end{pgfonlayer}
  \begin{pgfonlayer}{edgelayer}
    \draw [\commandkey{edgestyle}] (0.center) to (1);
  \end{pgfonlayer}
\end{tikzpicture}\,}

\newcommand{\bistateadj}[1]{\,\begin{tikzpicture}[yshift=-3mm]
  \begin{pgfonlayer}{nodelayer}
    \node [style=kpointadj, minimum width=1 cm, inner sep=2pt] (0) at (0, 0.5) {$\psi$};
    \node [style=none] (1) at (-0.75, 0.25) {};
    \node [style=none] (2) at (0.75, 0.25) {};
    \node [style=none] (3) at (-0.75, -0.5) {};
    \node [style=none] (4) at (0.75, -0.5) {};
  \end{pgfonlayer}
  \begin{pgfonlayer}{edgelayer}
    \draw (1.center) to (3.center);
    \draw (2.center) to (4.center);
  \end{pgfonlayer}
\end{tikzpicture}\,}

\newkeycommand{\pointketbra}[style1=copoint,style2=point][2]{\,%
\begin{tikzpicture}
  \begin{pgfonlayer}{nodelayer}
    \node [style=none] (0) at (0, -1.5) {};
    \node [style=\commandkey{style1}] (1) at (0, -0.75) {$#1$};
    \node [style=\commandkey{style2}] (2) at (0, 0.75) {$#2$};
    \node [style=none] (3) at (0, 1.5) {};
  \end{pgfonlayer}
  \begin{pgfonlayer}{edgelayer}
    \draw (0.center) to (1);
    \draw (2) to (3.center);
  \end{pgfonlayer}
\end{tikzpicture}\,}

\newkeycommand{\twocopointketbra}[style1=copoint,style2=copoint,style3=point][3]{\,%
\begin{tikzpicture}
  \begin{pgfonlayer}{nodelayer}
    \node [style=none] (0) at (-0.75, -1.5) {};
    \node [style=none] (1) at (0.75, -1.5) {};
    \node [style=\commandkey{style1}] (2) at (-0.75, -0.75) {$#1$};
    \node [style=\commandkey{style2}] (3) at (0.75, -0.75) {$#2$};
    \node [style=\commandkey{style3}] (4) at (0, 0.75) {$#3$};
    \node [style=none] (5) at (0, 1.5) {};
  \end{pgfonlayer}
  \begin{pgfonlayer}{edgelayer}
    \draw (0.center) to (2);
    \draw (1.center) to (3);
    \draw (4) to (5.center);
  \end{pgfonlayer}
\end{tikzpicture}\,}

\newkeycommand{\twopointketbra}[style1=copoint,style2=point,style3=point][3]{\,%
\begin{tikzpicture}
  \begin{pgfonlayer}{nodelayer}
    \node [style=none] (0) at (0, -1.5) {};
    \node [style=none] (1) at (-0.75, 1.5) {};
    \node [style=\commandkey{style1}] (2) at (0, -0.75) {$#1$};
    \node [style=\commandkey{style2}] (3) at (-0.75, 0.75) {$#2$};
    \node [style=\commandkey{style3}] (4) at (0.75, 0.75) {$#3$};
    \node [style=none] (5) at (0.75, 1.5) {};
  \end{pgfonlayer}
  \begin{pgfonlayer}{edgelayer}
    \draw (0.center) to (2);
    \draw (1.center) to (3);
    \draw (4) to (5.center);
  \end{pgfonlayer}
\end{tikzpicture}\,}

\newkeycommand{\pointbraket}[style1=point,style2=copoint][2]{\,%
\begin{tikzpicture}
  \begin{pgfonlayer}{nodelayer}
    \node [style=\commandkey{style1}] (0) at (0, -0.5) {$#1$};
    \node [style=\commandkey{style2}] (1) at (0, 0.5) {$#2$};
  \end{pgfonlayer}
  \begin{pgfonlayer}{edgelayer}
    \draw (0) to (1);
  \end{pgfonlayer}
\end{tikzpicture}\,}

\newkeycommand{\kpointbraket}[style1=kpoint,style2=kpoint adjoint][2]{\,%
\begin{tikzpicture}
  \begin{pgfonlayer}{nodelayer}
    \node [style=\commandkey{style1}] (0) at (0, -0.75) {$#1$};
    \node [style=\commandkey{style2}] (1) at (0, 0.75) {$#2$};
  \end{pgfonlayer}
  \begin{pgfonlayer}{edgelayer}
    \draw (0) to (1);
  \end{pgfonlayer}
\end{tikzpicture}\,}

\newkeycommand{\kpointmap}[style1=kpoint,style2=map][2]{\,%
\begin{tikzpicture}
  \begin{pgfonlayer}{nodelayer}
    \node [style=\commandkey{style1}] (0) at (0, -1.1) {$#1$};
    \node [style=\commandkey{style2}] (1) at (0, 0.2) {$#2$};
    \node [style=none] (2) at (0, 1.5) {};
  \end{pgfonlayer}
  \begin{pgfonlayer}{edgelayer}
    \draw (0) to (1);
    \draw (1) to (2.center);
  \end{pgfonlayer}
\end{tikzpicture}%
\,}

\newkeycommand{\sandwichtwo}[style1=point,style2=map,style3=map,style4=copoint][4]{\,%
\begin{tikzpicture}
  \begin{pgfonlayer}{nodelayer}
    \node [style=\commandkey{style2}] (0) at (0, -0.75) {$#2$};
    \node [style=\commandkey{style3}] (1) at (0, 0.75) {$#3$};
    \node [style=\commandkey{style1}] (2) at (0, -2) {$#1$};
    \node [style=\commandkey{style4}] (3) at (0, 2) {$#4$};
  \end{pgfonlayer}
  \begin{pgfonlayer}{edgelayer}
    \draw (2) to (0);
    \draw (0) to (1);
    \draw (1) to (3);
  \end{pgfonlayer}
\end{tikzpicture}\,}

\newkeycommand{\longbraket}[style1=point,style2=copoint][2]{\,%
\begin{tikzpicture}
  \begin{pgfonlayer}{nodelayer}
    \node [style=\commandkey{style1}] (0) at (0, -2) {$#1$};
    \node [style=\commandkey{style2}] (1) at (0, 2) {$#2$};
  \end{pgfonlayer}
  \begin{pgfonlayer}{edgelayer}
    \draw (0) to (1);
  \end{pgfonlayer}
\end{tikzpicture}\,}

\newkeycommand{\onb}[style=point]{\ensuremath{\left\{\tikz{\node[style=\commandkey{style}] (x) at (0,-0.3) {$j$};\draw (x)--(0,0.7);}\right\}}\xspace}

\newkeycommand{\copointmap}[style=copoint][1]{\,\tikz{\node[style=\commandkey{style}] (x) at (0,0.3) {$#1$};\draw (x)--(0,-0.7);}\,}

\tikzstyle{cdiag}=[matrix of math nodes, row sep=3em, column sep=3em, text height=1.5ex, text depth=0.25ex,inner sep=0.5em]
\tikzstyle{arrow above}=[transform canvas={yshift=0.5ex}]
\tikzstyle{arrow below}=[transform canvas={yshift=-0.5ex}]

\usetikzlibrary{shapes.multipart}

\tikzstyle{xor}=[circuit ee IEC, bulb,fill=white,rotate=45]
\tikzstyle{merge}=[circuit ee IEC, bulb,fill=white]
\tikzstyle{dmerge}=[circuit ee IEC, bulb,fill=black!20]
\tikzset{->-/.style={decoration={
  markings,
  mark=at position .5 with {\arrow{>}}},postaction={decorate}}}
\tikzset{-<-/.style={decoration={
  markings,
  mark=at position .5 with {\arrow{<}}},postaction={decorate}}}

\tikzstyle{bwSpider}=[
       rectangle split,
       rectangle split parts=2,
       rectangle split part fill={black,white},
 minimum size=3.6 mm, inner sep=-2mm, draw=black,scale=0.5,rounded corners=0.8 mm
       ]
 \tikzstyle{wbSpider}=[
       rectangle split,
       rectangle split parts=2,
       rectangle split part fill={white,black},
 minimum size=3.6 mm, inner sep=-2mm, draw=black,scale=0.5,rounded corners=0.8 mm
       ]

\tikzstyle{epiCopoint}=[regular polygon,regular polygon sides=3,draw,scale=0.75,inner sep=-0.5pt,minimum width=5mm,fill=white,regular polygon rotate=0,line width=1pt]
\tikzstyle{epiPoint}=[regular polygon,regular polygon sides=3,draw,scale=0.75,inner sep=-0.5pt,minimum width=5mm,fill=white,regular polygon rotate=180,line width=1pt]
\tikzstyle{epiPointWide}=[regular polygon,regular polygon sides=3,draw,scale=0.75,inner sep=-0.5pt,minimum width=8mm,fill=white,regular polygon rotate=180,line width=1pt]
\tikzstyle{epiBox}=[fill=white,draw, line width = 1pt,inner sep=0.6mm,font=\footnotesize,minimum height=3mm,minimum width=3mm]
\tikzstyle{epiBoxWide}=[fill=white,draw, line width = 1pt,inner sep=0.6mm,font=\footnotesize,minimum height=3mm,minimum width=5mm]
\tikzstyle{epiBoxVeryWide}=[fill=white,draw, line width = 1pt,inner sep=0.6mm,font=\footnotesize,minimum height=3mm,minimum width=7mm]
\tikzstyle{bWire}=[line width = 1.1pt, color=violet]
\tikzstyle{oWire}=[line width = .5pt, color=black!50, double, double distance = 1pt]
\tikzstyle{osWire}=[line width = .9pt, color=black!40]
\tikzstyle{qWire}=[line width = 1pt, color=black]
\tikzstyle{cWire}=[color=gray,line width = .75pt]
\tikzstyle{CqWire}=[color=gray,line width = .75pt,->-]
\tikzstyle{CcWire}=[color=gray,line width = .75pt,->-]
\tikzstyle{RqWire}=[line width = 1pt, color=black,-<-]
\tikzstyle{RcWire}=[color=gray,line width = .75pt,-<-]
\tikzstyle{env}=[copoint,regular polygon rotate=0,minimum width=0.2cm, fill=black]

\tikzstyle{probs}=[shape=semicircle,fill=white,draw=black,shape border rotate=180,minimum width=1.2cm]

\tikzstyle{every picture}=[baseline=-0.25em,scale=0.5]
\tikzstyle{dotpic}=[] 
\tikzstyle{diredges}=[every to/.style={diredge}]
\tikzstyle{math matrix}=[matrix of math nodes,left delimiter=(,right delimiter=),inner sep=2pt,column sep=1em,row sep=0.5em,nodes={inner sep=0pt},text height=1.5ex, text depth=0.25ex]

\tikzstyle{inline text}=[text height=1.5ex, text depth=0.25ex,yshift=0.5mm]
\tikzstyle{label}=[font=\footnotesize,text height=1.5ex, text depth=0.25ex,yshift=0.5mm]
\tikzstyle{left label}=[label,anchor=east,xshift=1.5mm]
\tikzstyle{right label}=[label,anchor=west,xshift=-1mm]
\tikzstyle{up label}=[label,anchor=south,yshift=-1mm]

\tikzstyle{braceedge}=[decorate,decoration={brace,amplitude=2mm,raise=-1mm}]
\tikzstyle{small braceedge}=[decorate,decoration={brace,amplitude=1mm,raise=-1mm}]

\tikzstyle{doubled}=[line width=1.6pt] 
\tikzstyle{boldedge}=[doubled,shorten <=-0.17mm,shorten >=-0.17mm]
\tikzstyle{boldedgegray}=[doubled,gray,shorten <=-0.17mm,shorten >=-0.17mm]
\tikzstyle{singleedgegray}=[gray]

\tikzstyle{semidoubled}=[line width=1.4pt] 
\tikzstyle{semiboldedgegray}=[semidoubled,gray,shorten <=-0.17mm,shorten >=-0.17mm]

\tikzstyle{boxedge}=[semiboldedgegray]

\tikzstyle{boldedgedashed}=[very thick,dashed,shorten <=-0.17mm,shorten >=-0.17mm]
\tikzstyle{vboldedgedashed}=[doubled,dashed,shorten <=-0.17mm,shorten >=-0.17mm]
\tikzstyle{left hook arrow}=[left hook-latex]
\tikzstyle{right hook arrow}=[right hook-latex]
\tikzstyle{sembracket}=[line width=0.5pt,shorten <=-0.07mm,shorten >=-0.07mm]

\tikzstyle{causal edge}=[->,thick,gray]
\tikzstyle{causal nondir}=[thick,gray]
\tikzstyle{timeline}=[thick,gray, dashed]

\tikzstyle{cedge}=[<->,thick,gray!70!white]

\tikzstyle{empty diagram}=[draw=gray!40!white,dashed,shape=rectangle,minimum width=1cm,minimum height=1cm]
\tikzstyle{empty diagram small}=[draw=gray!50!white,dashed,shape=rectangle,minimum width=0.6cm,minimum height=0.5cm]

\tikzstyle{dot}=[inner sep=0mm,minimum width=2mm,minimum height=2mm,draw,shape=circle]
\tikzstyle{bigdot}=[inner sep=0mm,minimum width=5mm,minimum height=5mm,draw,shape=circle]
\tikzstyle{leak}=[white dot, shape=regular polygon, minimum size=3.3 mm, regular polygon sides=3, outer sep=-0.2mm, regular polygon rotate=270]
\tikzstyle{proj}=[regular polygon,regular polygon sides=4,draw,scale=0.75,inner sep=-0.5pt,minimum width=6mm,fill=white]
\tikzstyle{projOut}=[regular polygon,regular polygon sides=3,draw,scale=0.75,inner sep=-0.5pt,minimum width=7.5mm,fill=white,regular polygon rotate=180]
\tikzstyle{projIn}=[regular polygon,regular polygon sides=3,draw,scale=0.75,inner sep=-0.5pt,minimum width=7.5mm,fill=white]
\tikzstyle{Vleak}=[white dot, shape=regular polygon, minimum size=3.3 mm, regular polygon sides=3, outer sep=-0.2mm, regular polygon rotate=90]
\tikzstyle{dleak}=[white dot, line width=1.6pt, shape=regular polygon, minimum size=3.3 mm, regular polygon sides=3, outer sep=-0.2mm, regular polygon rotate=270]

\tikzstyle{Wsquare}=[white dot, shape=regular polygon, rounded corners=0.8 mm, minimum size=3.3 mm, regular polygon sides=3, outer sep=-0.2mm]
\tikzstyle{Wsquareadj}=[white dot, shape=regular polygon, rounded corners=0.8 mm, minimum size=3.3 mm, regular polygon sides=3, outer sep=-0.2mm, regular polygon rotate=180]
\tikzstyle{ddot}=[inner sep=0mm, doubled, minimum width=2.5mm,minimum height=2.5mm,draw,shape=circle]

\tikzstyle{clear dot}=[dot,fill=none,text depth=-0.2mm,draw=gray, line width = .75pt]
\tikzstyle{tall clear dot}=[dot,fill=none,text depth=-0.2mm,draw=gray, line width = .75pt,shape=ellipse, minimum height=5mm]
\tikzstyle{wide clear dot}=[dot,fill=none,text depth=-0.2mm,draw=gray, line width = .75pt, shape=ellipse, minimum width = 5mm]
\tikzstyle{very wide clear dot}=[dot,fill=none,text depth=-0.2mm,draw=gray, line width = .75pt, shape=ellipse, minimum width = 7mm ]

\tikzstyle{black dot}=[dot,fill=black]
\tikzstyle{white dot}=[dot,fill=white,,text depth=-0.2mm]
\tikzstyle{white Wsquare}=[Wsquare,fill=gray,,text depth=-0.2mm]
\tikzstyle{white Wsquareadj}=[Wsquareadj,fill=white,,text depth=-0.2mm]
\tikzstyle{green dot}=[white dot] 
\tikzstyle{gray dot}=[dot,fill=gray!40!white,,text depth=-0.2mm]
\tikzstyle{red dot}=[gray dot] 

\tikzstyle{black ddot}=[ddot,fill=black]
\tikzstyle{white ddot}=[ddot,fill=white]
\tikzstyle{gray ddot}=[ddot,fill=gray!40!white]

\tikzstyle{gray edge}=[gray!60!white]

\tikzstyle{small dot}=[inner sep=0.2mm,minimum width=0pt,minimum height=0pt,draw,shape=circle]

\tikzstyle{small black dot}=[small dot,fill=black]
\tikzstyle{small white dot}=[small dot,fill=white]
\tikzstyle{small gray dot}=[small dot,fill=gray,draw=gray]

\tikzstyle{causal dot}=[inner sep=0.4mm,minimum width=0pt,minimum height=0pt,draw=white,shape=circle,fill=gray!40!white]

\tikzstyle{phase dimensions}=[minimum size=5mm,font=\footnotesize,rectangle,rounded corners=2.5mm,inner sep=0.2mm,outer sep=-2mm]
\tikzstyle{dphase dimensions}=[minimum size=5mm,font=\footnotesize,rectangle,rounded corners=2.5mm,inner sep=0.2mm,outer sep=-2mm]

\tikzstyle{white phase dot}=[dot,fill=white,phase dimensions]
\tikzstyle{white phase ddot}=[ddot,fill=white,dphase dimensions]

\tikzstyle{white rect ddot}=[draw=black,fill=white,doubled,minimum size=5mm,font=\footnotesize,rectangle,rounded corners=2.5mm,inner sep=0.2mm]
\tikzstyle{gray rect ddot}=[draw=black,fill=gray!40!white,doubled,minimum size=6mm,font=\footnotesize,rectangle,rounded corners=3mm]

\tikzstyle{gray phase dot}=[dot,fill=gray!40!white,phase dimensions]
\tikzstyle{gray phase ddot}=[ddot,fill=gray!40!white,dphase dimensions]
\tikzstyle{grey phase dot}=[gray phase dot]
\tikzstyle{grey phase ddot}=[gray phase ddot]

\tikzstyle{small phase dimensions}=[minimum size=4mm,font=\tiny,rectangle,rounded corners=2mm,inner sep=0.2mm,outer sep=-2mm]
\tikzstyle{small dphase dimensions}=[minimum size=4mm,font=\tiny,rectangle,rounded corners=2mm,inner sep=0.2mm,outer sep=-2mm]

\tikzstyle{small gray phase dot}=[dot,fill=gray!40!white,small phase dimensions]
\tikzstyle{small gray phase ddot}=[ddot,fill=gray!40!white,small dphase dimensions]

\tikzstyle{small map}=[draw,shape=rectangle,minimum height=4mm,minimum width=4mm,fill=white]

\tikzstyle{cnot}=[fill=white,shape=circle,inner sep=-1.4pt]

\tikzstyle{asym hadamard}=[fill=white,draw,shape=NEbox,inner sep=0.6mm,font=\footnotesize,minimum height=4mm]
\tikzstyle{asym hadamard conj}=[fill=white,draw,shape=NWbox,inner sep=0.6mm,font=\footnotesize,minimum height=4mm]
\tikzstyle{asym hadamard dag}=[fill=white,draw,shape=SEbox,inner sep=0.6mm,font=\footnotesize,minimum height=4mm]

\tikzstyle{hadamard}=[fill=white,draw,inner sep=0.6mm,font=\footnotesize,minimum height=4mm,minimum width=4mm]
\tikzstyle{small hadamard}=[fill=white,draw,inner sep=0.6mm,minimum height=1.5mm,minimum width=1.5mm]
\tikzstyle{small hadamard rotate}=[small hadamard,rotate=45]
\tikzstyle{dhadamard}=[hadamard,doubled]
\tikzstyle{small dhadamard}=[small hadamard,doubled]
\tikzstyle{small dhadamard rotate}=[small hadamard rotate,doubled]
\tikzstyle{antipode}=[white dot,inner sep=0.3mm,font=\footnotesize]

\tikzstyle{scalar}=[diamond,draw,inner sep=0.5pt,font=\small]
\tikzstyle{dscalar}=[diamond,doubled, draw,inner sep=0.5pt,font=\small]

\tikzstyle{small box}=[rectangle,inline text,fill=white,draw,minimum height=5mm,yshift=-0.5mm,minimum width=5mm,font=\small]
\tikzstyle{small gray box}=[small box,fill=gray!30]
\tikzstyle{medium box}=[rectangle,inline text,fill=white,draw,minimum height=5mm,yshift=-0.5mm,minimum width=10mm,font=\small]
\tikzstyle{square box}=[small box] 
\tikzstyle{medium gray box}=[small box,fill=gray!30]
\tikzstyle{semilarge box}=[rectangle,inline text,fill=white,draw,minimum height=5mm,yshift=-0.5mm,minimum width=12.5mm,font=\small]
\tikzstyle{large box}=[rectangle,inline text,fill=white,draw,minimum height=5mm,yshift=-0.5mm,minimum width=15mm,font=\small]
\tikzstyle{large gray box}=[small box,fill=gray!30]

\tikzstyle{Bayes box}=[rectangle,fill=black,draw, minimum height=3mm, minimum width=3mm]

\tikzstyle{gray square point}=[small box,fill=gray!50]

\tikzstyle{dphase box white}=[dhadamard]
\tikzstyle{dphase box gray}=[dhadamard,fill=gray!50!white]
\tikzstyle{phase box white}=[hadamard]
\tikzstyle{phase box gray}=[hadamard,fill=gray!50!white]

\tikzstyle{point}=[regular polygon,regular polygon sides=3,draw,scale=0.75,inner sep=-0.5pt,minimum width=9mm,fill=white,regular polygon rotate=180]
\tikzstyle{infpoint}=[regular polygon,regular polygon sides=3,draw,scale=0.75,inner sep=-0.5pt,minimum width=9mm,fill=white,regular polygon rotate=90]
\tikzstyle{point nosep}=[regular polygon,regular polygon sides=3,draw,scale=0.75,inner sep=-2pt,minimum width=9mm,fill=white,regular polygon rotate=180]
\tikzstyle{infcopoint}=[regular polygon,regular polygon sides=3,draw,scale=0.75,inner sep=-0.5pt,minimum width=9mm,fill=white,regular polygon rotate=270]
\tikzstyle{copoint}=[regular polygon,regular polygon sides=3,draw,scale=0.75,inner sep=-0.5pt,minimum width=9mm,fill=white]
\tikzstyle{dpoint}=[point,doubled]
\tikzstyle{dcopoint}=[copoint,doubled]

\tikzstyle{pointgrow}=[shape=cornerpoint,kpoint common,scale=0.75,inner sep=3pt]
\tikzstyle{pointgrow dag}=[shape=cornercopoint,kpoint common,scale=0.75,inner sep=3pt]

\tikzstyle{wide copoint}=[fill=white,draw,shape=isosceles triangle,shape border rotate=90,isosceles triangle stretches=true,inner sep=0pt,minimum width=1.5cm,minimum height=6.12mm]
\tikzstyle{wide point}=[fill=white,draw,shape=isosceles triangle,shape border rotate=-90,isosceles triangle stretches=true,inner sep=0pt,minimum width=1.5cm,minimum height=6.12mm,yshift=-0.0mm]
\tikzstyle{wide point plus}=[fill=white,draw,shape=isosceles triangle,shape border rotate=-90,isosceles triangle stretches=true,inner sep=0pt,minimum width=1.74cm,minimum height=7mm,yshift=-0.0mm]

\tikzstyle{wide dpoint}=[fill=white,doubled,draw,shape=isosceles triangle,shape border rotate=-90,isosceles triangle stretches=true,inner sep=0pt,minimum width=1.5cm,minimum height=6.12mm,yshift=-0.0mm]

\tikzstyle{tinypoint}=[regular polygon,regular polygon sides=3,draw,scale=0.55,inner sep=-0.15pt,minimum width=6mm,fill=white,regular polygon rotate=180]

\tikzstyle{white point}=[point]
\tikzstyle{white dpoint}=[dpoint]
\tikzstyle{green point}=[white point] 
\tikzstyle{white copoint}=[copoint]
\tikzstyle{gray point}=[point,fill=gray!40!white]
\tikzstyle{gray dpoint}=[gray point,doubled]
\tikzstyle{red point}=[gray point] 
\tikzstyle{gray copoint}=[copoint,fill=gray!40!white]
\tikzstyle{gray dcopoint}=[gray copoint,doubled]

\tikzstyle{white point guide}=[regular polygon,regular polygon sides=3,font=\scriptsize,draw,scale=0.65,inner sep=-0.5pt,minimum width=9mm,fill=white,regular polygon rotate=180]

\tikzstyle{black point}=[point,fill=black,font=\color{white}]
\tikzstyle{black copoint}=[copoint,fill=black,font=\color{white}]

\tikzstyle{tiny gray point}=[tinypoint,fill=gray!40!white]

\tikzstyle{diredge}=[->]
\tikzstyle{ddiredge}=[<->]
\tikzstyle{rdiredge}=[<-]
\tikzstyle{thickdiredge}=[->, very thick]
\tikzstyle{pointer edge}=[->,very thick,gray]
\tikzstyle{pointer edge part}=[very thick,gray]
\tikzstyle{dashed edge}=[dashed]
\tikzstyle{thick dashed edge}=[very thick,dashed]
\tikzstyle{thick gray dashed edge}=[thick dashed edge,gray!40]
\tikzstyle{thick map edge}=[very thick,|->]

\makeatletter
\newcommand{\boxshape}[3]{%
\pgfdeclareshape{#1}{
\inheritsavedanchors[from=rectangle] 
\inheritanchorborder[from=rectangle]
\inheritanchor[from=rectangle]{center}
\inheritanchor[from=rectangle]{north}
\inheritanchor[from=rectangle]{south}
\inheritanchor[from=rectangle]{west}
\inheritanchor[from=rectangle]{east}
\backgroundpath{
\southwest \pgf@xa=\pgf@x \pgf@ya=\pgf@y
\northeast \pgf@xb=\pgf@x \pgf@yb=\pgf@y

\@tempdima=#2
\@tempdimb=#3

\pgfpathmoveto{\pgfpoint{\pgf@xa - 5pt + \@tempdima}{\pgf@ya}}
\pgfpathlineto{\pgfpoint{\pgf@xa - 5pt - \@tempdima}{\pgf@yb}}
\pgfpathlineto{\pgfpoint{\pgf@xb + 5pt + \@tempdimb}{\pgf@yb}}
\pgfpathlineto{\pgfpoint{\pgf@xb + 5pt - \@tempdimb}{\pgf@ya}}
\pgfpathlineto{\pgfpoint{\pgf@xa - 5pt + \@tempdima}{\pgf@ya}}
\pgfpathclose
}
}}

\boxshape{NEbox}{0pt}{3pt}
\boxshape{SEbox}{0pt}{-3pt}
\boxshape{NWbox}{5pt}{0pt}
\boxshape{SWbox}{-5pt}{0pt}
\boxshape{EBox}{-3pt}{3pt}
\boxshape{WBox}{3pt}{-3pt}
\makeatother

\tikzstyle{cloud}=[shape=cloud,draw,minimum width=1.5cm,minimum height=1.5cm]

\tikzstyle{map}=[draw,shape=NEbox,inner sep=1pt,minimum height=4mm,fill=white]
\tikzstyle{dashedmap}=[draw,dashed,shape=NEbox,inner sep=2pt,minimum height=6mm,fill=white]
\tikzstyle{mapdag}=[draw,shape=SEbox,inner sep=1pt,minimum height=4mm,fill=white]
\tikzstyle{mapadj}=[draw,shape=SEbox,inner sep=2pt,minimum height=6mm,fill=white]
\tikzstyle{maptrans}=[draw,shape=SWbox,inner sep=2pt,minimum height=6mm,fill=white]
\tikzstyle{mapconj}=[draw,shape=NWbox,inner sep=2pt,minimum height=6mm,fill=white]

\tikzstyle{medium map}=[draw,shape=NEbox,inner sep=2pt,minimum height=6mm,fill=white,minimum width=7mm]
\tikzstyle{medium map dag}=[draw,shape=SEbox,inner sep=2pt,minimum height=6mm,fill=white,minimum width=7mm]
\tikzstyle{medium map adj}=[draw,shape=SEbox,inner sep=2pt,minimum height=6mm,fill=white,minimum width=7mm]
\tikzstyle{medium map trans}=[draw,shape=SWbox,inner sep=2pt,minimum height=6mm,fill=white,minimum width=7mm]
\tikzstyle{medium map conj}=[draw,shape=NWbox,inner sep=2pt,minimum height=6mm,fill=white,minimum width=7mm]
\tikzstyle{semilarge map}=[draw,shape=NEbox,inner sep=2pt,minimum height=6mm,fill=white,minimum width=9.5mm]
\tikzstyle{semilarge map trans}=[draw,shape=SWbox,inner sep=2pt,minimum height=6mm,fill=white,minimum width=9.5mm]
\tikzstyle{semilarge map adj}=[draw,shape=SEbox,inner sep=2pt,minimum height=6mm,fill=white,minimum width=9.5mm]
\tikzstyle{semilarge map dag}=[draw,shape=SEbox,inner sep=2pt,minimum height=6mm,fill=white,minimum width=9.5mm]
\tikzstyle{semilarge map conj}=[draw,shape=NWbox,inner sep=2pt,minimum height=6mm,fill=white,minimum width=9.5mm]
\tikzstyle{large map}=[draw,shape=NEbox,inner sep=2pt,minimum height=6mm,fill=white,minimum width=12mm]
\tikzstyle{large map conj}=[draw,shape=NWbox,inner sep=2pt,minimum height=6mm,fill=white,minimum width=12mm]
\tikzstyle{very large map}=[draw,shape=NEbox,inner sep=2pt,minimum height=6mm,fill=white,minimum width=17mm]

\tikzstyle{medium dmap}=[draw,doubled,shape=NEbox,inner sep=2pt,minimum height=6mm,fill=white,minimum width=7mm]
\tikzstyle{medium dmap dag}=[draw,doubled,shape=SEbox,inner sep=2pt,minimum height=6mm,fill=white,minimum width=7mm]
\tikzstyle{medium dmap adj}=[draw,doubled,shape=SEbox,inner sep=2pt,minimum height=6mm,fill=white,minimum width=7mm]
\tikzstyle{medium dmap trans}=[draw,doubled,shape=SWbox,inner sep=2pt,minimum height=6mm,fill=white,minimum width=7mm]
\tikzstyle{medium dmap conj}=[draw,doubled,shape=NWbox,inner sep=2pt,minimum height=6mm,fill=white,minimum width=7mm]
\tikzstyle{semilarge dmap}=[draw,doubled,shape=NEbox,inner sep=2pt,minimum height=6mm,fill=white,minimum width=9.5mm]
\tikzstyle{semilarge dmap trans}=[draw,doubled,shape=SWbox,inner sep=2pt,minimum height=6mm,fill=white,minimum width=9.5mm]
\tikzstyle{semilarge dmap adj}=[draw,doubled,shape=SEbox,inner sep=2pt,minimum height=6mm,fill=white,minimum width=9.5mm]
\tikzstyle{semilarge dmap dag}=[draw,doubled,shape=SEbox,inner sep=2pt,minimum height=6mm,fill=white,minimum width=9.5mm]
\tikzstyle{semilarge dmap conj}=[draw,doubled,shape=NWbox,inner sep=2pt,minimum height=6mm,fill=white,minimum width=9.5mm]
\tikzstyle{large dmap}=[draw,doubled,shape=NEbox,inner sep=2pt,minimum height=6mm,fill=white,minimum width=12mm]
\tikzstyle{large dmap conj}=[draw,doubled,shape=NWbox,inner sep=2pt,minimum height=6mm,fill=white,minimum width=12mm]
\tikzstyle{large dmap trans}=[draw,doubled,shape=SWbox,inner sep=2pt,minimum height=6mm,fill=white,minimum width=12mm]
\tikzstyle{large dmap adj}=[draw,doubled,shape=SEbox,inner sep=2pt,minimum height=6mm,fill=white,minimum width=12mm]
\tikzstyle{large dmap dag}=[draw,doubled,shape=SEbox,inner sep=2pt,minimum height=6mm,fill=white,minimum width=12mm]
\tikzstyle{very large dmap}=[draw,doubled,shape=NEbox,inner sep=2pt,minimum height=6mm,fill=white,minimum width=19.5mm]

\tikzstyle{muxbox}=[draw,shape=rectangle,minimum height=3mm,minimum width=3mm,fill=white]
\tikzstyle{dmuxbox}=[muxbox,doubled]

\tikzstyle{box}=[draw,shape=rectangle,inner sep=2pt,minimum height=6mm,minimum width=6mm,fill=white]
\tikzstyle{dbox}=[draw,doubled,shape=rectangle,inner sep=2pt,minimum height=6mm,minimum width=6mm,fill=white]
\tikzstyle{dmap}=[draw,doubled,shape=NEbox,inner sep=2pt,minimum height=6mm,fill=white]
\tikzstyle{dmapdag}=[draw,doubled,shape=SEbox,inner sep=2pt,minimum height=6mm,fill=white]
\tikzstyle{dmapadj}=[draw,doubled,shape=SEbox,inner sep=2pt,minimum height=6mm,fill=white]
\tikzstyle{dmaptrans}=[draw,doubled,shape=SWbox,inner sep=2pt,minimum height=6mm,fill=white]
\tikzstyle{dmapconj}=[draw,doubled,shape=NWbox,inner sep=2pt,minimum height=6mm,fill=white]

\tikzstyle{ddmap}=[draw,doubled,dashed,shape=NEbox,inner sep=2pt,minimum height=6mm,fill=white]
\tikzstyle{ddmapdag}=[draw,doubled,dashed,shape=SEbox,inner sep=2pt,minimum height=6mm,fill=white]
\tikzstyle{ddmapadj}=[draw,doubled,dashed,shape=SEbox,inner sep=2pt,minimum height=6mm,fill=white]
\tikzstyle{ddmaptrans}=[draw,doubled,dashed,shape=SWbox,inner sep=2pt,minimum height=6mm,fill=white]
\tikzstyle{ddmapconj}=[draw,doubled,dashed,shape=NWbox,inner sep=2pt,minimum height=6mm,fill=white]

\boxshape{sNEbox}{0pt}{3pt}
\boxshape{sSEbox}{0pt}{-3pt}
\boxshape{sNWbox}{3pt}{0pt}
\boxshape{sSWbox}{-3pt}{0pt}
\tikzstyle{smap}=[draw,shape=sNEbox,fill=white]
\tikzstyle{smapdag}=[draw,shape=sSEbox,fill=white]
\tikzstyle{smapadj}=[draw,shape=sSEbox,fill=white]
\tikzstyle{smaptrans}=[draw,shape=sSWbox,fill=white]
\tikzstyle{smapconj}=[draw,shape=sNWbox,fill=white]

\tikzstyle{dsmap}=[draw,dashed,shape=sNEbox,fill=white]
\tikzstyle{dsmapdag}=[draw,dashed,shape=sSEbox,fill=white]
\tikzstyle{dsmaptrans}=[draw,dashed,shape=sSWbox,fill=white]
\tikzstyle{dsmapconj}=[draw,dashed,shape=sNWbox,fill=white]

\boxshape{mNEbox}{0pt}{10pt}
\boxshape{mSEbox}{0pt}{-10pt}
\boxshape{mNWbox}{10pt}{0pt}
\boxshape{mSWbox}{-10pt}{0pt}
\tikzstyle{mmap}=[draw,shape=mNEbox]
\tikzstyle{mmapdag}=[draw,shape=mSEbox]
\tikzstyle{mmaptrans}=[draw,shape=mSWbox]
\tikzstyle{mmapconj}=[draw,shape=mNWbox]

\tikzstyle{mmapgray}=[draw,fill=gray!40!white,shape=mNEbox]
\tikzstyle{smapgray}=[draw,fill=gray!40!white,shape=sNEbox]

\makeatletter

\pgfdeclareshape{cornerpoint}{
\inheritsavedanchors[from=rectangle] 
\inheritanchorborder[from=rectangle]
\inheritanchor[from=rectangle]{center}
\inheritanchor[from=rectangle]{north}
\inheritanchor[from=rectangle]{south}
\inheritanchor[from=rectangle]{west}
\inheritanchor[from=rectangle]{east}
\backgroundpath{
\southwest \pgf@xa=\pgf@x \pgf@ya=\pgf@y
\northeast \pgf@xb=\pgf@x \pgf@yb=\pgf@y

\pgfmathsetmacro{\pgf@shorten@left}{\pgfkeysvalueof{/tikz/shorten left}}
\pgfmathsetmacro{\pgf@shorten@right}{\pgfkeysvalueof{/tikz/shorten right}}

\pgfpathmoveto{\pgfpoint{0.5 * (\pgf@xa + \pgf@xb)}{\pgf@ya - 5pt}}
\pgfpathlineto{\pgfpoint{\pgf@xa - 8pt + \pgf@shorten@left}{\pgf@yb - 1.5 * \pgf@shorten@left}}
\pgfpathlineto{\pgfpoint{\pgf@xa - 8pt + \pgf@shorten@left}{\pgf@yb}}
\pgfpathlineto{\pgfpoint{\pgf@xb + 8pt - \pgf@shorten@right}{\pgf@yb}}
\pgfpathlineto{\pgfpoint{\pgf@xb + 8pt - \pgf@shorten@right}{\pgf@yb - 1.5 * \pgf@shorten@right}}
\pgfpathclose
}
}

\pgfdeclareshape{cornercopoint}{
\inheritsavedanchors[from=rectangle] 
\inheritanchorborder[from=rectangle]
\inheritanchor[from=rectangle]{center}
\inheritanchor[from=rectangle]{north}
\inheritanchor[from=rectangle]{south}
\inheritanchor[from=rectangle]{west}
\inheritanchor[from=rectangle]{east}
\backgroundpath{
\southwest \pgf@xa=\pgf@x \pgf@ya=\pgf@y
\northeast \pgf@xb=\pgf@x \pgf@yb=\pgf@y

\pgfmathsetmacro{\pgf@shorten@left}{\pgfkeysvalueof{/tikz/shorten left}}
\pgfmathsetmacro{\pgf@shorten@right}{\pgfkeysvalueof{/tikz/shorten right}}

\pgfpathmoveto{\pgfpoint{0.5 * (\pgf@xa + \pgf@xb)}{\pgf@yb + 5pt}}
\pgfpathlineto{\pgfpoint{\pgf@xa - 8pt + \pgf@shorten@left}{\pgf@ya + 1.5 * \pgf@shorten@left}}
\pgfpathlineto{\pgfpoint{\pgf@xa - 8pt + \pgf@shorten@left}{\pgf@ya}}
\pgfpathlineto{\pgfpoint{\pgf@xb + 8pt - \pgf@shorten@right}{\pgf@ya}}
\pgfpathlineto{\pgfpoint{\pgf@xb + 8pt - \pgf@shorten@right}{\pgf@ya + 1.5 * \pgf@shorten@right}}
\pgfpathclose
}
}

\makeatother

\pgfkeyssetvalue{/tikz/shorten left}{0pt}
\pgfkeyssetvalue{/tikz/shorten right}{0pt}

\tikzstyle{kpoint common}=[draw,fill=white,inner sep=1pt,minimum height=4mm]
\tikzstyle{kpoint sc}=[shape=cornerpoint,kpoint common]
\tikzstyle{kpoint adjoint sc}=[shape=cornercopoint,kpoint common]
\tikzstyle{kpoint}=[shape=cornerpoint,shorten left=5pt,kpoint common]
\tikzstyle{kpoint adjoint}=[shape=cornercopoint,shorten left=5pt,kpoint common]
\tikzstyle{kpoint conjugate}=[shape=cornerpoint,shorten right=5pt,kpoint common]
\tikzstyle{kpoint transpose}=[shape=cornercopoint,shorten right=5pt,kpoint common]
\tikzstyle{kpoint symm}=[shape=cornerpoint,shorten left=5pt,shorten right=5pt,kpoint common]

\tikzstyle{wide kpoint sc}=[shape=cornerpoint,kpoint common, minimum width=1 cm]
\tikzstyle{wide kpointdag sc}=[shape=cornercopoint,kpoint common, minimum width=1 cm]

\tikzstyle{black kpoint}=[shape=cornerpoint,shorten left=5pt,kpoint common,fill=black,font=\color{white}]

\tikzstyle{black kpoint sm}=[shape=cornerpoint,shorten left=5pt,kpoint common,fill=black,font=\color{white},scale=0.75]

\tikzstyle{black kpoint adjoint}=[shape=cornercopoint,shorten left=5pt,kpoint common,fill=black,font=\color{white}]
\tikzstyle{black kpointadj}=[shape=cornercopoint,shorten left=5pt,kpoint common,fill=black,font=\color{white}]

\tikzstyle{black kpointadj sm}=[shape=cornercopoint,shorten left=5pt,kpoint common,fill=black,font=\color{white},scale=0.75]

\tikzstyle{black dkpoint}=[shape=cornerpoint,shorten left=5pt,kpoint common,fill=black, doubled,font=\color{white}]
\tikzstyle{black dkpoint adjoint}=[shape=cornercopoint,shorten left=5pt,kpoint common,fill=black, doubled,font=\color{white}]
\tikzstyle{black dkpointadj}=[shape=cornercopoint,shorten left=5pt,kpoint common,fill=black, doubled,font=\color{white}]

\tikzstyle{black dkpoint sm}=[shape=cornerpoint,shorten left=5pt,kpoint common,fill=black, doubled,font=\color{white},scale=0.75]
\tikzstyle{black dkpointadj sm}=[shape=cornercopoint,shorten left=5pt,kpoint common,fill=black, doubled,font=\color{white},scale=0.75]

\tikzstyle{kpointdag}=[kpoint adjoint]
\tikzstyle{kpointadj}=[kpoint adjoint]
\tikzstyle{kpointconj}=[kpoint conjugate]
\tikzstyle{kpointtrans}=[kpoint transpose]

\tikzstyle{big kpoint}=[kpoint, minimum width=1.2 cm, minimum height=8mm, inner sep=4pt, text depth=3mm]

\tikzstyle{wide kpoint}=[kpoint, minimum width=1 cm, inner sep=2pt]
\tikzstyle{wide kpointdag}=[kpointdag, minimum width=1 cm, inner sep=2pt]
\tikzstyle{wide kpointconj}=[kpointconj, minimum width=1 cm, inner sep=2pt]
\tikzstyle{wide kpointtrans}=[kpointtrans, minimum width=1 cm, inner sep=2pt]

\tikzstyle{wider kpoint}=[kpoint, minimum width=1.25 cm, inner sep=2pt]
\tikzstyle{wider kpointdag}=[kpointdag, minimum width=1.25 cm, inner sep=2pt]
\tikzstyle{wider kpointconj}=[kpointconj, minimum width=1.25 cm, inner sep=2pt]
\tikzstyle{wider kpointtrans}=[kpointtrans, minimum width=1.25 cm, inner sep=2pt]

\tikzstyle{gray kpoint}=[kpoint,fill=gray!50!white]
\tikzstyle{gray kpointdag}=[kpointdag,fill=gray!50!white]
\tikzstyle{gray kpointadj}=[kpointadj,fill=gray!50!white]
\tikzstyle{gray kpointconj}=[kpointconj,fill=gray!50!white]
\tikzstyle{gray kpointtrans}=[kpointtrans,fill=gray!50!white]

\tikzstyle{gray dkpoint}=[kpoint,fill=gray!50!white,doubled]
\tikzstyle{gray dkpointdag}=[kpointdag,fill=gray!50!white,doubled]
\tikzstyle{gray dkpointadj}=[kpointadj,fill=gray!50!white,doubled]
\tikzstyle{gray dkpointconj}=[kpointconj,fill=gray!50!white,doubled]
\tikzstyle{gray dkpointtrans}=[kpointtrans,fill=gray!50!white,doubled]

\tikzstyle{white label}=[draw,fill=white,rectangle,inner sep=0.7 mm]
\tikzstyle{gray label}=[draw,fill=gray!50!white,rectangle,inner sep=0.7 mm]
\tikzstyle{black label}=[draw,fill=black,rectangle,inner sep=0.7 mm]

\tikzstyle{dkpoint}=[kpoint,doubled]
\tikzstyle{wide dkpoint}=[wide kpoint,doubled]
\tikzstyle{dkpointdag}=[kpoint adjoint,doubled]
\tikzstyle{wide dkpointdag}=[wide kpointdag,doubled]
\tikzstyle{dkcopoint}=[kpoint adjoint,doubled]
\tikzstyle{dkpointadj}=[kpoint adjoint,doubled]
\tikzstyle{dkpointconj}=[kpoint conjugate,doubled]
\tikzstyle{dkpointtrans}=[kpoint transpose,doubled]

\tikzstyle{kscalar}=[kpoint common, shape=EBox, inner xsep=-1pt, inner ysep=3pt,font=\small]
\tikzstyle{kscalarconj}=[kpoint common, shape=WBox, inner xsep=-1pt, inner ysep=3pt,font=\small]

\tikzstyle{spekpoint}=[kpoint sc,minimum height=5mm,inner sep=3pt]
\tikzstyle{spekcopoint}=[kpoint adjoint sc,minimum height=5mm,inner sep=3pt]

\tikzstyle{dspekpoint}=[spekpoint,doubled]
\tikzstyle{dspekcopoint}=[spekcopoint,doubled]

 \tikzstyle{upground}=[circuit ee IEC,thick,ground,rotate=90,scale=2.5]
 \tikzstyle{downground}=[circuit ee IEC,thick,ground,rotate=-90,scale=2.5]
 \tikzstyle{infupground}=[circuit ee IEC,thick,ground,rotate=0,scale=2.5]
 \tikzstyle{infdownground}=[circuit ee IEC,thick,ground,rotate=180,scale=2.5]
 \tikzstyle{bigground}=[regular polygon,regular polygon sides=3,draw=gray,scale=0.50,inner sep=-0.5pt,minimum width=10mm,fill=gray]

\tikzstyle{arrs}=[-latex,font=\small,auto]
\tikzstyle{arrow plain}=[arrs]
\tikzstyle{arrow dashed}=[dashed,arrs]
\tikzstyle{arrow bold}=[very thick,arrs]
\tikzstyle{arrow hide}=[draw=white!0,-]
\tikzstyle{arrow reverse}=[latex-]
\tikzstyle{cdnode}=[]

\tikzstyle{tilde}=[draw=blue]
\tikzstyle{tildelabel}=[text=blue]

\let\olddagger\dagger
\renewcommand{\dagger}{\ensuremath{\olddagger}\xspace}

\usepackage{makeidx}
\makeindex

\theoremstyle{plain}
\newtheorem*{main theorem}{Main Theorem}
\newtheorem{theorem}{Theorem}[section]

\newtheorem{proposition}[theorem]{Proposition}

\newtheorem{definition}[theorem]{Definition}

\newtheorem{example*}[theorem]{Example*}
\newtheorem{examples*}[theorem]{Examples*}
\newtheorem{remark}[theorem]{Remark}
\newtheorem{remark*}[theorem]{Remark*}

\newtheorem*{search problem}{Search Problem}

\usepackage{color}
\def\bR{\begin{color}{red}}
\def\bB{\begin{color}{blue}}
\def\bM{\begin{color}{magenta}}
\def\bC{\begin{color}{cyan}}
\def\bW{\begin{color}{white}}
\def\bBl{\begin{color}{black}}
\def\bG{\begin{color}{green}}
\def\bY{\begin{color}{yellow}}
\def\e{\end{color}\xspace}
\newcommand{\bit}{\begin{itemize}}
\newcommand{\eit}{\end{itemize}\par\noindent}
\newcommand{\ben}{\begin{enumerate}}
\newcommand{\een}{\end{enumerate}\par\noindent}
\newcommand{\beq}{\begin{equation}}
\newcommand{\eeq}{\end{equation}\par\noindent}
\newcommand{\beqa}{\begin{eqnarray*}}
\newcommand{\eeqa}{\end{eqnarray*}\par\noindent}
\newcommand{\beqn}{\begin{eqnarray}}
\newcommand{\eeqn}{\end{eqnarray}\par\noindent}

\def\jR{\begin{color}{black}}
\def\jB{\begin{color}{black}}
\def\jM{\begin{color}{magenta}}
\def\jC{\begin{color}{cyan}}
\def\jW{\begin{color}{white}}
\def\jBl{\begin{color}{black}}
\def\jG{\begin{color}{green}}
\def\jY{\begin{color}{yellow}}

\newcommand{\xiNC}{\xi_{\rm{\kern -0.8pt n \kern -0.7pt c}}}
\usepackage{soul}

\newcommand\stoptoc{\let\addcontentsline\nocontentsline}
\newcommand\resumetoc{\let\addcontentsline\origcontentsline}

\makeatletter
\let\origsubsubsection\subsubsection
\renewcommand{\subsubsection}[1]{%
	\begingroup
	\let\addcontentsline\@gobblethree 
	\origsubsubsection{#1}
	\endgroup
}
\makeatother

\begin{document}

\title{\fontsize{19pt}{0pt}\selectfont\hspace{-0.4mm}Decoupling local classicality from classical explainability:\ \newline A noncontextual model for bilocal classical theory and a~locally-classical but contextual theory}

\author{Sina Soltani}
\email{sina.soltani@phdstud.ug.edu.pl}
\affiliation{International Centre for Theory of Quantum Technologies, Uniwersytet Gdański, ul.~Jana Bażyńskiego 1A, 80-309 Gdańsk, Polska}
\author{Marco Erba~\raisebox{0.47ex}{\scalebox{2.1}{\orcidlink{0000-0002-2172-592X}}}}
\affiliation{International Centre for Theory of Quantum Technologies, Uniwersytet Gdański, ul.~Jana Bażyńskiego 1A, 80-309 Gdańsk, Polska}
\orcid{0000-0002-2172-592X}
\email{recobrama@gmail.com}
\author{David Schmid}
\affiliation{Perimeter Institute for Theoretical Physics, 31 Caroline Street North, Waterloo, Ontario Canada N2L 2Y5}
\author{John H.~Selby}
\affiliation{International Centre for Theory of Quantum Technologies, Uniwersytet Gdański, ul.~Jana Bażyńskiego 1A, 80-309 Gdańsk, Polska}
\affiliation{Theoretical Sciences Visiting Program, Okinawa Institute of Science and Technology Graduate University, Onna, 9040495, Japan}
\email{john.h.selby@gmail.com}

\maketitle

\begin{abstract}
We construct an ontological model for the theory known as bilocal classical theory~\cite{d2020classicality}. To our knowledge, this is only the second time that an ontological model has been constructed for an entire theory, rather than just for some particular scenarios within a theory. This result refutes a conjecture from Ref.~\cite{d2020classicality} which suggested that there might be no local-realist ontological model for bilocal classical theory. Moreover, it is the first time that an ontological model has been constructed for a theory that fails to be locally tomographic, showing that the assumption of local tomography underpinning the structure theorem in Ref.~\cite{schmid2024structuretheorem} is a genuine limitation of the theorem. This demonstrates that in general there is no tension between failures of local tomography and classical explainability (i.e., generalised noncontextuality). In fact, bilocal classical theory is in many ways more simply understood via the underlying ontological model than it is within its original formulation (much as how odd-dimensional stabiliser subtheories can be more simply understood via Spekkens' toy theory). Furthermore, this result naturally leads to the question, does every locally-classical theory admit of an ontological model? By constructing a concrete counterexample, we show that this is not the case. Our findings demonstrate that there is no straightforward relationship between theories being locally-classical, and them being classically-explainable. This shows that the fundamental status of compositional properties (such as local tomography) is not a technical side-issue, but a central and unavoidable question for a coherent understanding even of classicality itself.
\end{abstract}

	\tableofcontents

\section{Introduction}

The formalism of operational probabilistic theories (OPTs)~\cite{chiribella2010probabilistic,d2017quantum,Rolino_2025,soltani2025noncontextualontologicalmodelsoperational,erba2025categorical} (closely related to the one of generalised probabilistic theories, or GPTs for short~\cite{GPT_Barrett,PhysRevA.87.052131,Janotta_2014,lami2018nonclassicalcorrelationsquantummechanics,PLAVALA20231}) is a broad framework encompassing essentially all theories that admit of an operational characterisation. In this framework, classical probability theory (CT) has traditionally been singled out as the unique ``classical'' theory~\cite{GPT_Barrett,PhysRevA.87.052131,Janotta_2014,werner2014comment,PLAVALA20231,PhysRevLett.128.160402,Schmid2024reviewreformulation}. From an operational standpoint, this conception of classicality essentially boils down to three characteristic traits of systems in CT: \emph{joint perfect discriminability of pure states}---implying that every state space is geometrically described by a simplex; the  \emph{distinguishability of any two bipartite states by using the statistics of measurements on the individual components}---entailing, for simplicial systems, a composition rule that can be described by the Cartesian product; and that \emph{all mathematically consistent processes are physically possible} (``everything that is not forbidden is allowed''~\cite{d2017quantum}). The second feature operationally boils down to the \emph{local tomography}~\cite{Wooter_1990,hardy2012limited,centeno2024twirled,lismer2025experimentaltestprincipletomographic,baldijão2026tomographicallynonlocalentanglement} or \emph{local discriminability}~\cite{d2017quantum} principles, while the third is known as the \emph{no-restriction hypothesis}~\cite{chiribella2010probabilistic,PhysRevA.87.052131}, which are important properties also satisfied by quantum theory. Note that \emph{joint perfect discriminability of pure states} is a single-system property, while \emph{local tomography} is a multi-system, compositional notion. 
In brief, CT is the operational probabilistic theory that contains all and only processes consistent with simpliciality and local tomography.

One can also ask which theories are {\em classically-explainable}, in the sense that they can be reasonably explained by CT. This question was first answered in Refs.~\cite{SchmidGPT,schmid2024structuretheorem,schmid2020unscrambling,muller2023testing}, which argued that one GPT can explain the statistical, convex, and compositional structures of another if and only if there is a linear and diagram preserving map between the two. In particular, then, a GPT is classically-explainable if it can be {\em embedded} into CT in this manner. Such an embedding is called an ontological model of a GPT. This conception of classical explainability was subsequently explored in a number of works~\cite{selby2023accessible,selbylinear,soltani2025noncontextualontologicalmodelsoperational,Schmid2025shadowssubsystemsof,rossi2025typicalcontextuality}. 

This notion of classical explainability  coincides with the existence of a gen\-er\-al\-ised-non\-con\-tex\-tual~\cite{gencontext} explanation (for the operational theory which leads to the GPT~\cite{schmid2024structuretheorem}) within an ontological model~\cite{harrigan2010einstein}. The latter approach can also be motivated by a methodological version of Leibniz's principle~\cite{Leibniz}, by its equivalence to the existence of some positive quasiprobabilistic representation~\cite{negativity,schmid2024structuretheorem}, and by its equivalence to the above notion of classical explainability. Failures of generalised noncontextuality are strong manifestations of failures of classical explainability, and have been found (and often leveraged as resources) in the realms of quantum computation~\cite{Schmid2022Stabilizer,shahandeh2021quantum}, state discrimination~\cite{schmid2018contextual,flatt2021contextual,mukherjee2021discriminating,Shin2021}, interference~\cite{Catani2023whyinterference,catani2022reply,catani2023aspects,giordani2023experimental}, compatibility~\cite{selby2023incompatibility,selby2023accessible,PhysRevA.109.022239,PhysRevResearch.2.013011}, uncertainty relations~\cite{catani2022nonclassical}, metrology~\cite{contextmetrology}, thermodynamics~\cite{contextmetrology,comar2024contextuality,lostaglio2018}, weak values~\cite{AWV, KLP19}, coherence~\cite{rossi2023contextuality,Wagner2024coherence, wagner2024inequalities}, quantum Darwinism~\cite{baldijao2021noncontextuality}, information processing and communication~\cite{POM,RAC,RAC2,Saha_2019,Yadavalli2020,PhysRevLett.119.220402,fonseca2024robustness}, cloning~\cite{cloningcontext}, broadcasting~\cite{jokinen2024nobroadcasting}, pre- and post-selection paradoxes~\cite{PP1}, randomness certification~\cite{Roch2021}, psi-epistemicity~\cite{Leifer}, and Bell~\cite{Wright2023invertible,schmid2020unscrambling} and Kochen-Specker scenarios~\cite{operationalks,kunjwal2018from,Kunjwal16,Kunjwal19,Kunjwal20,specker,Gonda2018almostquantum}. 

Relatedly, some recent works~\cite{d2020classical,d2020classicality,PhysRevA.109.022239,Rolino_2025} on OPTs have explored the idea, and modifications thereof, of \emph{locally-classical} theories---that is, those theories where all systems are classical in the strict sense of CT (i.e., containing all and only processes consistent with simpliciality), while the system-composition rule may differ from that of standard CT (thus generally violating local tomography)~\cite{d2020classical,d2020classicality,erba2024compositionrulequantumsystems,centeno2024twirled}.\footnote{Note that, since our approach is \emph{compositional}, composite systems can be always taken to be themselves components of larger composites. Motivated by this observation, ``locally-classical'' here means that every system, be it elementary or composite, is classical---while, globally, the theory may differ from the standard classical theory, particularly in its composition rule. Such a terminology therefore slightly differs from the one used, e.g., in Refs.~\cite{PhysRevLett.104.140401,PhysRevLett.109.090403}, where the analysis was carried out for multipartite correlation scenarios and, accordingly, ``locally quantum'' refers to settings where the elementary components are quantum, while the global systems need not necessarily be quantum.\label{fn:locally}} Such theories can only differ from CT in their composition rule, which may not obey local tomography. The first example of a fully-fledged locally-classical theory that violates local tomography is known as \emph{bilocal classical theory}~\cite{d2020classicality}, which we introduce in Sec.~\ref{Sec:BCT}. 

In this paper, we further explore the connections between these different concepts relating to classicality, that is, between \emph{local classicality} and \emph{classical explainability} (in the precise senses described above). Our first result is to show that bilocal classical theory is classically-explainable. Hence, modifications to the composition rule of classical theory (as an OPT) are not necessarily incompatible with the existence of an underlying ontology which composes following the standard ``classical'' composition rule (the Cartesian product). On the other hand, we construct alternative examples of locally-classical theories (following the construction of \emph{latent theories} introduced in Ref.~\cite{erba2024compositionrulequantumsystems}) which we prove are {\em not} classically-explainable. Consequently, locally-classical theories that fail to be tomographically local are {\em sometimes, but not always,} classically-explainable. An important direction for future work would be to find some way to classify such theories, for example, by finding some physical principle which singles out which are classically-explainable.

\bigskip

 In the remainder of the introduction (and in more detail in App.~\ref{App:Prelim}), we provide a concise overview of the framework of operational probabilistic theories, tailored to the needs of our discussion.  In particular, defining \emph{ontological models} of OPTs~\cite{soltani2025noncontextualontologicalmodelsoperational}, and what it means for such a model to be \emph{classically-explainable}. In section 2 of the paper, we construct such a model for bilocal classical theory and show that the constructed model is a consistent one. Finally, in section 3, we conclude the paper with a no-go theorem concerning another type of locally-classical theories, termed \emph{latent classical theories}, for which we prove that no ontological model can be constructed.

\subsection{Operational probabilistic theories}\label{Sec:OPTs}

In this section, we briefly review the framework of operational probabilistic theories (OPTs). For a more detailed introduction, we refer readers to App.~\ref{app:OPTs-introduction} and references therein. The primitives in the OPT framework are systems, transformations,
and probabilities.\footnote{At least, these are the primitives of relevance for this paper.} A system is representative of a physical object under study in a laboratory, such as a particle, atom, or field. A transformation describes the evolution a system undergoes. The OPT framework prescribes the probabilities of different outcomes one would obtain in a measurement, given the preparations and evolutions that have occurred before the measurement, taking into account how these primitives are composed.

We denote the class of systems for a generic OPT $\Theta $ as $\textbf{Sys}(\Theta)$. For single systems, we use Latin characters $A, B, C, \cdots \in \textbf{Sys}(\Theta)$. We denote the class of transformations from system $A$ to $B$ as $\textbf{Transf}(A\to B)$, and we depict a transformation $t \in \textbf{Transf}(A\to B)$ as
\begin{equation}
	\input{Diagramss/36_event_t.tikz}\;.
\end{equation}

When not describing any physical system, we use the trivial system $I\in\textbf{Sys}(\Theta)$, which we represent it as a blank space. Thus, for a particular class of transformations $\rho\in \textbf{Transf}(I\to A)$ and  $a\in \textbf{Transf}(A\to I) $, we can depict them as
\begin{equation}
\input{Diagramss/40_preobs.tikz}\;.
\end{equation}
These are respectively called \emph{states} and \emph{effects}, where for a system $A$, we denote them as $\textbf{St}(A)$ and $\textbf{Eff}(A)$ respectively. Sometimes, we denote a state for a system $A$ as $\vert \rho )_A$, and an effect $a$ as $( a\vert_{A}$. We also denote the sequential composition of a state $\vert \rho )_A$ and an effect $( a\vert_{A}$ as $(a\vert \rho)_{A}$.

In the OPT framework, a transformation $p$ with both trivial input and output is called a \textit{scalar}, which we demand to be a probability, i.e., $p \in [0,1]$, which we depict as
\begin{equation}
    \input{Diagramss/62_prob2.tikz}\;.
\end{equation}
Thus, closed circuits in the OPT framework represent the probability that a given closed circuit occurs among all other possibilities for that particular closed circuit in a laboratory (see~\eqref{probability} and App.~\ref{app:OPTs-introduction} for a more detailed description of probabilities in OPTs).

 An OPT is given by composition rules, i.e. a sequential composition rule for transformations, and a parallel composition rule when transformations are combined in parallel to form another transformation. These are given by associative bilinear maps and must satisfy additional compatibility conditions to ensure well-defined composition in a given OPT~\cite{d2020classicality}. We denote sequential composition map as $\circ$ and parallel composition map as $\boxtimes$. Likewise, single systems can be composed together to form multipartite systems. 

It is well known~\cite{chiribella2010probabilistic,Chiribella_2014,chiribella2016quantum,d2017quantum} that one can define a complete class of linearly independent vectors in $\textbf{St}(A)$ that span a real vector space, denoted by $\textbf{St}_{\mathbb{R}}(A) := \text{Span}_{\mathbb{R}}\;\textbf{St}(A)$. Similarly, the space $\textbf{Eff}(A)$ spans, in general, a subspace of its dual vector space, defined as $\textbf{Eff}_{\mathbb{R}}(A) := \text{Span}_{\mathbb{R}}\;\textbf{Eff}(A)$. 
Similarly, transformations belong to a vector space, which allows us to define $\textbf{Transf}_{\mathbb{R}}(A \to B) := \text{Span}_{\mathbb{R}}\,\textbf{Transf}(A \to B)$. 
 A transformation $t \in \textbf{Transf}(A \to B)$ is called \textit{atomic} if, whenever one has $t = t_1 + t_2$ for some $t_1, t_2 \in \textbf{Transf}(A \to B)$, it follows that $t_1 \propto t_2$. That is, atomic transformations are those that, up to a scalar factor, cannot be further refined into other transformations. In the theories that we consider here, atomic transformations are important because every other transformation can be represented as a positive conic combination of atomic transformations. 

An OPT is said to satisfy \textit{local tomography}, if and only if, for any two distinct bipartite state $\rho\in \textbf{St}(AB)$ and $\rho'\in \textbf{St}(AB)$, there exist local effects  $ a_1\in\textbf{Eff}(A)$ and $a_2\in\textbf{Eff}(B)$, such that
\begin{equation}
\input{Diagramss/65_tom.tikz}\neq\input{Diagramss/65_tom1.tikz}\;.
\end{equation}
In other words, two bipartite states can be distinguished solely through local measurements.

The condition that the probability distribution associated with a preparation of a set of states does not depend on the specific choice of the measurement applied at its output is commonly referred to as \textit{causality}. It is well known that this notion of
causality is equivalent to the principle of no-signalling from future and equivalent to the statement that each system possesses a \textit{unique deterministic effect}~\cite{d2017quantum}. In particular, an OPT is said to be causal if and only if every system $A$ admits a unique deterministic effect, represented as 
\begin{equation}
\begin{tikzpicture}
	\begin{pgfonlayer}{nodelayer}
		\node [style=infupground] (20) at (0.75, 0) {};
		\node [style=none] (21) at (-0.75, 0) {};
		\node [style=none] (22) at (-0.75, 0) {};
		\node [style=none] (23) at (-0.25, 0.5) {\scriptsize$A$};
	\end{pgfonlayer}
	\begin{pgfonlayer}{edgelayer}
		\draw  (21.center) to (20);
	\end{pgfonlayer}
\end{tikzpicture}
\;.
\end{equation}
A transformation $t \in \textbf{Transf}(A \to B)$ is called \textit{deterministic} if and only if
\begin{equation}
\label{eq:normalisation}
\input{Diagramss/50_discard2.tikz}=\begin{tikzpicture}
	\begin{pgfonlayer}{nodelayer}
		\node [style=infupground] (20) at (0.75, 0) {};
		\node [style=none] (21) at (-0.75, 0) {};
		\node [style=none] (22) at (-0.75, 0) {};
		\node [style=none] (23) at (-0.25, 0.5) {\scriptsize$A$};
	\end{pgfonlayer}
	\begin{pgfonlayer}{edgelayer}
		\draw  (21.center) to (20);
	\end{pgfonlayer}
\end{tikzpicture}
\;.
\end{equation}
This means that the transformation occurs with certainty, i.e, with probability $1$. For deterministic states, this condition reduces to the normalisation condition. 

\subsection{Classical theory}\label{Sec:CT}

Classical theory is an  OPT $\Delta$ where systems are labelled by natural numbers $n \in \mathds{N}$ representing the number of perfectly distinguishable states (See~\eqref{eq:perfectly-discriminable} in App.~\ref{app:OPTs-introduction} for more detail) for the system. The composite of a system $n$ with a system $m$ gives us the system $nm$ (i.e., given my multiplication of integers). It can be shown that this composition rule ensures that classical theory satisfies local tomography.

Transformations from system $n$ to system $m$ are represented by substochastic matrices $M \in \textbf{SubStoch}(n \to m)$, while deterministic transformations correspond to stochastic matrices $S \in \textbf{Stoch}(n \to m)$.
As a special case, states are (subnormalised) probability distributions, effects correspond to response functions, and scalars correspond to probabilities. Sequential composition is given by the composition of (sub)stochastic matrices, whereas parallel composition is given by the usual tensor product. 
We depict systems in classical theory using grey wires
\begin{equation}
\begin{tikzpicture}
	\begin{pgfonlayer}{nodelayer}
		\node [style=none] (0) at (-2, 0) {};
		\node [style=none] (1) at (2, 0) {};
		\node [style=none] (2) at (0, 0.5) {$n$};
	\end{pgfonlayer}
	\begin{pgfonlayer}{edgelayer}
		\draw [osWire] (0.center) to (1.center);
	\end{pgfonlayer}
\end{tikzpicture}
\;.
\end{equation}
 We represent a pure state of a system $n$ in classical theory as 
\begin{equation}
\input{BCT-diagrams/17_CT1.tikz}\; ,
\end{equation}
where $1\leq i \leq n$.  Composing two pure states of system $n$ and $m$, reults in another pure state of the composite system $nm$ which we depict as
\begin{equation}
\input{BCT-diagrams/17_CT2.tikz}\;,
\end{equation}
Where $1\leq j\leq m$. A pure effect for a system $n$ in classical theory is depicted as
\begin{equation}
\input{BCT-diagrams/17_CT13.tikz}\;,
\end{equation}
where one has
\begin{equation}
(i'\vert i )=\delta_{i,i'},
\end{equation}
for all $1\leq i \leq n$, and all $1\leq i' \leq n$.

An atomic transformation from a system $n$ to $m$ is given by
\begin{equation}
\input{BCT-diagrams/17_CT4.tikz}:=\lambda \cdot \input{BCT-diagrams/17_CT3.tikz}\;,
\end{equation}
for some $1\leq i_0 \leq n, 1\leq l\leq m$ and $\lambda\in [0,1]$. Similarly, for composite systems, atomic transformations are given by
\begin{equation}
\input{BCT-diagrams/17_CT5.tikz}:=\lambda\cdot\input{BCT-diagrams/17_CT6.tikz}\;,
\end{equation}
where for lower indices, one has
$1 \leq i_1 \leq n_1$, $1 \leq i_2 \leq n_2$, $\ldots$, $1 \leq i_p \leq n_p$, and for upper indices, one has
$1 \leq l_1 \leq m_1$, $1 \leq l_2 \leq m_2$, $\ldots$, $1 \leq l_q \leq m_q$, and, $\lambda\in [0,1]$. 

We call atomic transformations with $\lambda=1$ \textit{normalised atomic transformations} which we denote as $\mathscr{A}_{i_0}^l$. Every arbitrary transformation  $t$ in classical theory, from a system $n$ to $m$, can be written as the unique conical combination of normalised atomic transformations, that is
\begin{equation}
\input{BCT-diagrams/17_CT7.tikz}=\sum_{i_0,l}C(i_0,l)\input{BCT-diagrams/17_CT8.tikz}\;,
\end{equation}
where $C(i_0,l)$ denotes the conical coefficient.\footnote{Note that, since classical theory satisfies local tomography, there is no need to include ancillary systems (see~\eqref{eq:local-tomography-no-ancilla} and App.~\ref{app:OPTs-introduction}) in order to fully characterise transformations. Consequently, atomic transformations can be described as specific instances of general transformations without reference to any ancillary system, as we have done in this analysis.}

\subsection{Ontological models and classical explainability}\label{Sec:ONT} 
In this section, we are interested in a particular notion of a map between two OPTs that helps us formulate our notion of classical explainability of OPTs. For a more detailed discussion of the notion of a map between two OPTs, we refer readers to App.~\ref{app:GNC}.
\begin{definition}[Ontological models of OPTs]\label{def:ONT}
An ontological model of an OPT $\input{Diagramss/47_ontmodelunq.tikz}$, depicted as 
\begin{equation}
	\tilde{\xi}::\input{Diagramss/47_ontmodelunq1.tikz}\mapsto \input{Diagramss/47_ontmodelunq2.tikz}\;\;,
\end{equation}
is a linear map from an OPT $\Theta$ to a classical theory $\Delta$ that satisfies diagram preservation and determinacy preservation. See App.~\ref{app:GNC} for a discussion of these properties, and Remarks~\ref{rem:ont-models-structure} and~\ref{rem:ont-models-structure2} for a more detailed discussion of the structure of ontological models.
\end{definition} 
An ontological model of an OPT associates each transformation with a substochastic transformation in classical theory. Additionally, it assigns a classical state to every state and a classical effect to every effect. Each system $A$ in the OPT is mapped to a classical system whose dimension, denoted by $\Lambda_{A}$, specifies the size of the corresponding ontic space~\cite{harrigan2010einstein}. Note that an important property an ontological model must satisfy is empirical adequacy, namely, preservation of scalars. In the present setting, this follows immediately from linearity and determinism preservation. Indeed, determinism preservation implies $\tilde{\xi}(1)=1$. Hence, by linearity,
$\tilde{\xi}(p)=\tilde{\xi}(p\cdot 1)=p\tilde{\xi}(1)=p$, for every $p\in [0,1]$, which is pictorially represented as 
\begin{equation}
\input{Diagramss/59_pp1.tikz}=\input{Diagramss/59_pp4.tikz}\;.
\end{equation}

If the constructed model is consistent and qualifies as a valid ontological model, then, as explained in~\cite{schmid2024structuretheorem,soltani2025noncontextualontologicalmodelsoperational}, it confirms that the OPT of interest is classically-explainable and compatible with the notion of generalised noncontextuality. Note that as the ontological model is diagram preserving, that this ensures that the causal structure of the ontological model is the same as the causal structure of the OPT, and in particular, means that such ontological models are Bell local. I.e., if an OPT admits of an ontological model then it can never violate any Bell inequality in any causal scenario. 

In Sec.~\ref{Sec:BCT}, we carry out this construction for BCT, focusing in particular on how the model maps systems, pure states, pure effects, and normalised atomic transformations. We then demonstrate the model’s consistency, thereby confirming the existence of an ontological model for BCT, and thus its classical explainability. In Sec.~\ref{Sec:LQT} we show that no such construction can exist for LCT, and hence, LCT is not classically explainable.

\section{A locally-classical theory that is classically-explainable}\label{Sec:BCT}

\subsection{Bilocal classical theory}
Bilocal classical theory (BCT) was first introduced in~\cite{d2020classicality} as a consistent theory of locally-classical systems with a composition rule that differs from that of classical theory--but in such a way that the state and effect spaces of composite systems are also classical (identical to that of the simplicial theory). Despite the fact that all systems are themselves classical, the exotic composition rule in BCT leads it to exhibit several features absent in classical theory, such as entanglement, cloning, entanglement swapping, dense coding, additivity of classical capacities, non-monogamous entanglement, and hypersignaling. In this section, we briefly review the basic properties of BCT that are most relevant to our purposes. In the following sections, we aim to construct an ontological model for BCT. To that end, we will focus in particular on the composition rule, the characterisation of pure states and pure effects, and atomic transformations in BCT. We do not aim to prove all of the properties we mention here. Interested readers can find full proofs and more discussions in Ref.~\cite{d2020classicality}.

Bilocal classical theory (BCT) is  an OPT in which every system is classical. Hence, we represent each single system by a natural number $n$ which for the trivial system, one has $n=1$. The main difference from classical theory lies in the composition rule. For any two single systems $n$ and $m$, the composite system, denoted as $n\boxtimes m$, is given by
\begin{equation}
\label{dimension-rule-BCT}
 n\boxtimes m =
\begin{cases}
	2 n m, & \text{if } n \ne 1 \ne m, \\
	n, & \text{if } m = 1,
    \\
	 m, & \text{if } n = 1.
\end{cases}
\end{equation}
In the $n\neq 1\neq m$ case, pure states are given as the vertices of a $2nm$-simplex. In light of the composition rule, BCT does not exhibit local tomography; instead, it exhibits bilocal tomography, so that it requires bipartite measurements to fully characterise multipartite states (hence why it has been named as bilocal classical theory).
We represent a system $n$ in BCT pictorially as as  thick violet wire 
\begin{equation}
\begin{tikzpicture}
	\begin{pgfonlayer}{nodelayer}
		\node [style=none] (0) at (-2, 0) {};
		\node [style=none] (1) at (2, 0) {};
		\node [style=none] (2) at (0, 0.5) {$n$};
	\end{pgfonlayer}
	\begin{pgfonlayer}{edgelayer}
		\draw [bWire] (0.center) to (1.center);
	\end{pgfonlayer}
\end{tikzpicture}
\;.
\end{equation}
A pure states of a single system $n$ can be represented as
\begin{equation}
\input{BCT-diagrams/13_BCT.tikz}\; ,
\end{equation}
where $1\leq i\leq n$. For a composite system $2nm$, with subsystems $n$ and $m$, we represent a pure state as
\begin{equation}
\scalebox{1.0}{\input{BCT-diagrams/13_BCT1.tikz}}\; ,
\end{equation}
where $1 \leq i \leq n$ and $1 \leq j \leq m$, with $s$ being a binary variable $s \in \lbrace 0, 1 \rbrace$. We can also represent a bipartite pure state as $\left\vert (ij)_{s} \right)$. For all pure states $\vert i )$ of a system $n$, and all pure states $\vert j )$ of a system $m$, the parallel composition is given by
\begin{equation}
\input{BCT-diagrams/13_BCT2.tikz}=\frac{1}{2}\sum_{s=0,1}\input{BCT-diagrams/13_BCT1.tikz}\; .
\end{equation}
This composition rule for pure states implies that when two pure states are composed in parallel, the resulting state—unlike in classical theory—is no longer pure. However, despite this difference, composite systems in BCT remain classical systems. Consequently, every state in BCT can still be expressed as a unique subconvex combination of pure states.

Similarly, a pure effect of a single system $n$ can be represented as
\begin{equation}
\input{BCT-diagrams/13_BCT15.tikz}\;,
\end{equation}
where $1\leq i\leq n$. For a bipartite system $2nm$, with subsystems $n$ and $m$, a pure effect is denoted as
\begin{equation}
\input{BCT-diagrams/13_BCT17.tikz}\;,
\end{equation}
where $1\leq i\leq n$, $1\leq j\leq m$ and $s\in\lbrace 0,1\rbrace$, which we can also denote it as $( (ij)_{s}\vert$. For all possible pure states of a single system $n$ and a bipartite system $2nm$—that is, for all $1 \leq i \leq n$, $1 \leq j \leq m$, and $s \in \lbrace 0,1\rbrace$—one has
\begin{align}
(  i^{'} \vert i ) &= \delta_{i^{'},i}\; ,\nonumber\\
(( i^{'}j^{'} )_{s^{'}} \vert (ij)_{s})&=\delta_{i^{'},i}\delta_{j^{'},j}\delta_{s^{'},s}\; ,
\end{align}
for all possible pure effects, i.e., for all $1 \leq i' \leq n$, $1 \leq j' \leq m$, and $s' \in \lbrace 0,1\rbrace$.

For an arbitrary composite system consists of systems $n_1$,$n_2$,$n_3$,$\ldots$,$n_p $, we represent a pure state as
\begin{equation}
\input{BCT-diagrams/13_BCT11.tikz}\; ,
\end{equation}
where $1\leq i_1\leq n_1$, $1\leq i_2\leq n_2$, $\ldots$, $1\leq i_p\leq n_P$, and $s_1$, $s_2$, $\ldots$, $s_{p-1}\in\lbrace 0,1\rbrace$. For a composite pure effect, one has
\begin{equation}
\input{BCT-diagrams/13_BCT16.tikz}\;,
\end{equation}
where $1\leq i_1\leq n_1$, $1\leq i_2\leq n_2$, $\ldots$, $1\leq i_p\leq n_P$, and $s_1$, $s_2$, $\ldots$, $s_{p-1}\in\lbrace 0,1\rbrace$. For a discussion on how to relate the multi-system and single-system definitions of states, effects, and transformations see App.~\ref{App:MultiVsSingle}. 

A transformation $\mathscr{A}$ from system $n$ to $m$ is \emph{atomic} if and only if, for every system $E$ used as an ancillary system, one has 
\begin{equation}
\label{atomic-trans}
\input{BCT-diagrams/13_BCT7.tikz}=\lambda\delta_{i,i_{0}}\input{BCT-diagrams/13_BCT8.tikz}\; ,
\end{equation}
for some $\lambda\in\left[0,1\right]$, $1\leq i_{0}\leq n$, $1\leq l\leq m $ and $\tau\in\lbrace 0,1\rbrace$, with the summation taken modulo two (see Ref.~\cite[Prop.~3]{d2020classicality} for the proof). Hence, we can denote an atomic transformation from system $n$ to $m$ as
\begin{equation}
\input{BCT-diagrams/13_BCT9.tikz}\;.
\end{equation}
Similarly to classical theory, we call atomic transformations with non-trivial systems with $\lambda=1$ normalised atomic transformations, which we denote as
\begin{equation}
\input{BCT-diagrams/15_ea2.tikz}\;,
\end{equation}
where one has
\begin{equation}
\label{atomic-channels}
\input{BCT-diagrams/13_BCT9.tikz}=\lambda\cdot \input{BCT-diagrams/15_ea2.tikz}\;.
\end{equation} 
An atomic transformation, with a composite system input and a composite system output, is represented as
\begin{equation}
\label{atomic-comp}
\input{BCT-diagrams/13_BCT12.tikz}\; ,
\end{equation}
where for lower indices, one has
$1 \leq i_1 \leq n_1$, $1 \leq i_2 \leq n_2$, $\ldots$, $1 \leq i_p \leq n_p$,
and $s_1$, $s_2$, $\ldots$, $s_{p-1} \in \lbrace 0,1\rbrace$,
and for upper indices, one has
$1 \leq l_1 \leq m_1$, $1 \leq l_2 \leq m_2$, $\ldots$, $1 \leq l_q \leq m_q$,
and $t_1$, $t_2$, $\ldots$, $t_{q-1} \in \lbrace 0,1\rbrace$.
Also, one has $\lambda \in [0,1]$ and $\tau \in \lbrace 0,1\rbrace$. More details on how to characterise atomic transformations on composite systems can be found in App.~\ref{atomiccomp}. It can be shown that, similarly to classical theory, sequential composition of atomic transformation yields another atomic transformation (see~\eqref{seq-comp-atomic} in App.~\ref{App:BCT}), However, in the case of parallel composition, there is a stark contrast between BCT and classical theory, that when atomic transformations are composed in parallel, the resulting transformation is no longer atomic~\cite{d2020classicality}. 

Similarly to classical theory we have that: 
\begin{proposition}
\label{uniqueness}
Every transformation $t$ from system $n$ to $m$ in BCT can be uniquely expressed as a conical combination of normalised atomic transformations:
\begin{equation}
\input{BCT-diagrams/15_ea.tikz}=\sum_{i_0,l,\tau}C(i_0,l,\tau)\cdot \input{BCT-diagrams/15_ea2.tikz}\;.
\end{equation}
\end{proposition}
\begin{proof}
 See Ref.~\cite[Prop.~3]{d2020classicality} for existence and App.~\ref{app:uniqueness-proof} for uniqueness.
\end{proof}

As a category, BCT is consistent in the sense that it satisfies the required coherence conditions. Moreover, as an OPT, it exhibits a consistent probabilistic structure. These properties have been discussed and demonstrated in Ref.~\cite{d2020classicality}.  We have also included a very brief introduction to \textit{reversible} transformations of BCT in App.~\ref{App:rev-trans-bct}\footnote{Note that, since BCT does not satisfy local tomography, unlike classical theory it is necessary to consider ancillary systems $E$ in order to fully characterise transformations (see~\eqref{eq:local-tomography-no-ancilla} and App.~\ref{app:OPTs-introduction}), as we have done in our analysis of atomic transformations in~\eqref{atomic-trans} (and reversible transformations in~\eqref{reversible-transformations} and App.~\ref{App:rev-trans}).}. 

\subsection{An ontological model for bilocal classical theory}

In this section, we construct an ontological model for BCT. But first, it will be useful for our purpose to introduce some notations regarding even dimensional classical theory. In classical theory, we can depict a single system of dimension $2n$ as
\begin{equation}
\begin{tikzpicture}
	\begin{pgfonlayer}{nodelayer}
		\node [style=none] (0) at (1.75, 0) {};
		\node [style=none] (1) at (-1.75, 0) {};
		\node [style=none] (2) at (0, 0.5) {$2n$};
	\end{pgfonlayer}
	\begin{pgfonlayer}{edgelayer}
		\draw [oWire] (1.center) to (0.center);
	\end{pgfonlayer}
\end{tikzpicture}
=\input{BCT-diagrams/1_SingleSys2.tikz}\;\;,
\end{equation}
wherein $n$ represents $n$ dimensional classical subspace, and $B$ represents a binary subspace. This representation of a single system enables us to represent a pure states of a $2n$ dimensional classical system in classical theory as
\begin{equation}
\input{BCT-diagrams/4_OntMapOfStJ2.tikz}\;\;,
\end{equation}
where $1\leq i \leq n$, and $s\in\lbrace 0,1\rbrace$. On the binary subspace, we denote the $t$-controlled NOT gate as
\begin{equation}
\input{BCT-diagrams/4_OntMapOfStJ3.tikz}=\input{BCT-diagrams/4_OntMapOfStJ4.tikz}=\input{BCT-diagrams/4_OntMapOfStJ26.tikz}\;\;, 
\end{equation} 
with the summation being modulo two. We use the symbol $\circ$ to denote the action of copying pure states. For instance, in the case of the binary subspace, applying the copying to the completely mixed state yields
\begin{equation}
\input{BCT-diagrams/4_OntMapOfStJ8.tikz}=\frac{1}{2}\left(\input{BCT-diagrams/4_OntMapOfStJ9.tikz}+\input{BCT-diagrams/4_OntMapOfStJ10.tikz}\right)\;\;.
\end{equation}
We can further generalize the copying operation on binary subspaces to accommodate multiple inputs and outputs as 
\beq
\input{BCT-diagrams/4_OntMapOfStJ13.tikz}:=\sum_s \input{BCT-diagrams/4_OntMapOfStJ14.tikz}
\eeq

Now, with these notations in place, we construct an ontological model for BCT as follows:

\begin{definition}[An ontological model for BCT]\label{BCT-ontological-model} The ontological map $\input{BCT-diagrams/2_OntMap.tikz}$ is the unique linear map specified by the following: 
\begin{itemize}
	\item $\tilde{\xi}$ assigns to each system $n$ an ontic space $\Lambda_{n}:=\tilde{\xi}{(n)}$, where for a non-trivial system $n\neq 1$, we define $\tilde{\xi}{(n)}$ as classical system with the dimension $2n$, and for the trivial system $n=1$, we define $\tilde{\xi}{(n)}$ to be the trivial system in classical theory as well, i.e., $\tilde{\xi}{(n)}=1$. For non-trivial systems, this is depicted as
	\begin{equation}
		\input{BCT-diagrams/1_SingleSys.tikz}:=\begin{tikzpicture}
	\begin{pgfonlayer}{nodelayer}
		\node [style=none] (6) at (1.75, 0) {};
		\node [style=none] (8) at (-1.75, 0) {};
		\node [style=none] (9) at (0, 0.5) {$\tilde{\xi}\left(n\right)$};
	\end{pgfonlayer}
	\begin{pgfonlayer}{edgelayer}
		\draw [oWire] (8.center) to (6.center);
	\end{pgfonlayer}
\end{tikzpicture}
\;,
	\end{equation}
	
	And for every composite system in BCT with non-trivial subsystems $n_{1}$, $n_{2}$, $n_{3}$ $\ldots$ $n_{p-1}, $ $n_{p}$, one has
	\begin{equation}
		\input{BCT-diagrams/2_CompSys.tikz}:=\input{BCT-diagrams/2_CompSys1.tikz}\;.
	\end{equation}
	\item For a pure state of a single system $n$, one has
	\begin{equation}
		\input{BCT-diagrams/1_1ParSt2.tikz}:=\input{BCT-diagrams/4_OntMapOfStJ1.tikz}\;,
	\end{equation}
	and for a pure state of a composite system in BCT with subsystems $n_{1}$, $n_{2}$, $n_{3}$ $\ldots$ $n_{p-1}, $ $n_{p}$, using not and copy gates, one has
	\begin{equation}
		\label{ont-model-on-composite-state}
		\input{BCT-diagrams/4_OntMapOfSt5.tikz}:=\input{BCT-diagrams/4_OntMapOfStJ.tikz}\;.
	\end{equation}
	\item For a pure effect of a single system $n$, one has
	\begin{equation}
		\label{effects}
		\input{BCT-diagrams/1_1ParSt5.tikz}:=\input{BCT-diagrams/10_effects.tikz}\;,
	\end{equation}
	and for a pure effects of a composite system, one has
	\begin{equation}
		\label{ont-model-on-composite-effect}
		\input{BCT-diagrams/10_effects1.tikz}:=\input{BCT-diagrams/10_effects2.tikz}\;.
	\end{equation}
	\item For a normalised atomic transformation with a single system $n$ as input, and a single system $m$ as output, one has
	\begin{equation}
		\label{ont-model-on-single-atomics}
		\input{BCT-diagrams/2_OntMapOfAtomTrans1.tikz}:=\input{BCT-diagrams/2_OntMapOfAtomTrans4.tikz}\;,
	\end{equation}
	and for a normalised atomic transformation in BCT, with composite systems as input and output, one has
	\begin{equation}
		\label{ont-model-on-composite-atomics}
		\scalebox{0.85}{\input{BCT-diagrams/3_CompAtomicTrans5.tikz}}\;\coloneqq\;\scalebox{0.85}{\input{BCT-diagrams/3_CompAtomicTrans8.tikz}}\;.
	\end{equation}
	\item Finally, we define the ontological model of BCT to be deterministic scalar-preserving, i.e., the ontological model maps scalar $1$ in BCT to scalar $1$ in classical theory, that is
	\begin{equation}
		\tilde{\xi}{(1)}:=1.
	\end{equation}
\end{itemize}
\end{definition}

We have defined the ontological model on single system entities, namely pure states, pure effects, and normalised atomic transformations, as well as their multipartite counterparts.  In App.~\ref{App:MultiVsSingle} we will see how these two views are related, where as an example, in App.~\ref{App:rev-trans-ont}, we use the results established in App.~\ref{App:MultiVsSingle} to formalize reversible transformations of BCT in the ontological model.

\subsubsection{Consistency checks}

In this section, we examine whether the ontological model constructed for BCT is consistent. By~\ref{def:ONT}, consistency requires the model to be linear, diagram preserving, probability preserving, or empirically adequate, and determinacy preserving. The model we have constructed is defined as a linear map, with its action specified on systems, pure states, pure effects, normalised atomic transformations, and the scalar $1$. It remains to verify that this definition is compatible with linearity and that the model satisfies diagram preservation, probability preservation, and determinacy preservation. We show these properties in what follows.

\

\paragraph{Linearity preservation}
We have defined the ontological model to be linear, but a priori this could lead to inconsistencies. Namely, if we had two different linear decompositions of the same transformations in terms of normalised atomic transformations, then the ontological model could conceivably depend on which decomposition we used. This, however, is not a problem for BCT because all of the normalised atomic transformations are linearly independent, as we showed in Prop.~\ref{uniqueness}. 

\paragraph{Diagram preservation}
Diagram preservation implies that for the sequential composition of two arbitrary transformations (for clarity, we omit explicit notation of systems and transformations), one has
\begin{equation}
\label{seqdp}
\input{BCT-diagrams/14_dp18.tikz}=\input{BCT-diagrams/14_dp19.tikz}\;.
\end{equation}
And for the parallel composition of two arbitrary transformations, one has
\begin{equation}
\label{pardp}
\input{BCT-diagrams/9_GenDiag.tikz}=\input{BCT-diagrams/9_GenDiag5.tikz}\;.
\end{equation}
One also has both identity and swap preservation (See~\eqref{identity-swap} in App.~\ref{dpproof}). When all of these conditions are satisfied, diagram preservation holds for every arbitrary scenario (like the one in~\eqref{dpexample}). 

The ontological model of BCT, as we have defined it, indeed satisfies diagram preservation, as we prove in App.~\ref{dpproof}. In doing so, we first show that diagram preservation holds for the sequential composition of normalised atomic transformations. Then, using the linearity of the ontological model and the fact that every transformation is uniquely given as a conical combination of normalised atomic transformations (by Proposition~\ref{uniqueness}), we conclude that the model is diagram preserving for arbitrary sequential compositions.

For parallel composition, we reduce the preservation of an arbitrary pair of transformations to simpler conditions (see~\eqref{dp2} in App.~\ref{dpproof})—including identity and swap preservation—whose joint satisfaction implies diagram preservation. These conditions are all fulfilled by virtue of Proposition~\ref{uniqueness} and the linearity of the ontological model.

Note that, for the sake of clarity, we have reviewed diagram preservation in Figures~\eqref{seqdp} and~\eqref{pardp} only for the case where all systems are single and non-trivial. In App.~\ref{dpproof}, we explain why diagram preservation also holds when systems are multipartite, not limited to this specific setting. 

\paragraph{Determinacy preservation}
The ontological model must map every deterministic transformation in BCT to a deterministic transformation in classical theory. This requirement is satisfied, as demonstrated in App.~\ref{detproof}, where we prove that the ontological model preserves the unique deterministic effect. By diagram preservation, this implies that the ontological model maps every deterministic transformation in BCT to a deterministic transformation in classical theory (i.e., a stochastic matrix). 

\subsection{Bilocal classical theory admits an ontological model}
We have shown that the model in Definition~\ref{BCT-ontological-model} satisfies all the requirements of Definition~\ref{def:ONT}, and so is a valid ontological model for BCT. Hence, we obtain the following result:
\begin{theorem}\label{thm:BCT-NCOM}
There exists an ontological model for bilocal classical theory.
\end{theorem}
Thus, BCT is classically explainable and compatible with the notion of generalised noncontextuality, in the sense discussed in~\cite{schmid2024structuretheorem,soltani2025noncontextualontologicalmodelsoperational}. This might lead one to suspect that every locally-classical OPT admits of an ontological model and is hence classically explainable. However, as we show in the next section, this is not the case. To demonstrate this, we leverage the general OPT construction presented in Ref.~\cite{erba2024compositionrulequantumsystems} to build a specific counterexample: a locally-classical theory that violates local tomography and does not admit an ontological model.

\section{A locally-classical theory that is not classically-explainable}\label{Sec:LQT}

\subsection{Latent classical theories}
Here, we aim to show that a class of locally-classical OPTs known as \emph{latent classical theories}—ob\-tain\-ed via the OPT construction introduced in Ref.~\cite{erba2024compositionrulequantumsystems}—unlike BCT, do not admit of any ontological model. This OPT construction was originally applied to quantum systems to define latent quantum theories (LQTs), namely, a family of OPTs that are locally quantum while violating local tomography.\footnote{``Locally quantum'' here means that every system, be it elementary or composite, is quantum---while the composition postulate may generally differ from the standard quantum one. See also footnote~\ref{fn:locally}.} By applying the same construction to classical systems, one can similarly define the following: 

\begin{definition}[Latent classical theories (LCTs)]
\emph{Latent classical theories} (LCTs) are the family of locally-classical OPTs defined by applying the construction of~\cite[Apps.~A-B]{erba2024compositionrulequantumsystems} to classical systems. Equivalently, LCTs may be defined as the restriction of LQTs~\cite{erba2024compositionrulequantumsystems} to their classical, namely, diagonal, sectors.

\end{definition}
The key feature of LCTs is that they generally differ from standard classical theory in the composition rule. Intuitively, the composition rule for latent classical system is the usual one up to extra degrees of freedom ``popping up'' in the form of factor spaces. This generalises the standard classical composition rule. More concretely, for every pair of elementary latent classical systems $C_1$ and $C_2$, there exists a classical space $L_{12}$ such that
\begin{equation*}
C_1C_2\coloneqq C_1\boxtimes C_2 = L_{12}\otimes C_1\otimes C_2.
\end{equation*}
The space $L_{12}$ is called the \emph{latent factor} associated with the composite $C_1C_2$. Moreover, there exists a normalised state $\left|\kappa\right)_{L_{12}}$, called the \emph{latent state} of $L_{12}$, such that the parallel composition of latent classical states reads as
\begin{equation}\label{eq:LCT-composite}
\left|\sigma\right)_{C_1}\boxtimes\left|\tau\right)_{C_2} = \left|\kappa\right)_{L_{12}}\otimes\left|\sigma\right)_{C_1}\otimes\left|\tau\right)_{C_2}\quad\forall
\left|\sigma\right)_{C_1}\in\textbf{St}(C_1),\left|\tau\right)_{C_2}\in\textbf{St}(C_2)
.
\end{equation}
Therefore, whenever a LCT has a latent factor that is non-trivial---i.e., that is not one-dimensional---the resulting LCT will be different from standard classical theory. With this, we can prove the following result.
\subsection{Latent classical theories do not admit ontological models}
\begin{theorem}\label{thm:counterexample}
There does not exist an ontological model for any latent classical theory  where at least one latent state is not full-rank.
\end{theorem}
\proof[Proof]
Let us consider any LCT where there exists a pair of elementary latent classical systems $C_1$ and $C_2$ whose associated latent factor $L_{12}$ has latent state $\left|\kappa\right)_{L_{12}}$ that is not full-rank. Note that this implies that the associated latent factor $L_{12}$ has dimension $d\geq2$, meaning that the LCTs considered here are all different from standard classical theory. Therefore, by construction, there also exists at least one non-null effect $\left(\kappa^{\perp}\right|_{L_{12}}$ such that $\left(\kappa^{\perp}|\kappa\right)_{L_{12}}=0$. Denoting the discarding of any system by $\left(\text{Tr}\right|$, we then define the composite effect $\left(b\right|_{C_1C_2}\coloneqq\left(\kappa^{\perp}\right|_{L_{12}}\otimes\left(\text{Tr}\right|_{C_1}\otimes\left(\text{Tr}\right|_{C_2}$ that, in spite of being non-null, annihilates every product state in the light of Eq.~\eqref{eq:LCT-composite}, i.e.,
\begin{equation}\label{eq:annihilating-effect}
\scalebox{0.9}{\input{Proof-latent-classical/locally-annihilating.tikz}}
\!=\ 
0
\qquad\forall
\left|\sigma\right)_{C_1}\in\textbf{St}(C_1),\left|\tau\right)_{C_2}\in\textbf{St}(C_2)
.
\end{equation}
The existence of such an effect is a non-trivial feature and is made possible by the violation of local tomography---as a necessary albeit not sufficient condition. Furthermore, one also has
\begin{equation}\label{eq:jellyfish}
\scalebox{0.9}{\input{Proof-latent-classical/jellyfish.tikz}}
\end{equation}
for all states $\left|\beta\right)_{C_1C_2}$ (where we recall that $\varepsilon$ denotes the null process). Eq.~\eqref{eq:jellyfish} can be easily verified by direct computation resorting to the general construction of Ref.~\cite[App.~B]{erba2024compositionrulequantumsystems} and using Eq.~\eqref{eq:annihilating-effect}. Let us define $\left|\bar{\kappa}\right)_{L_{12}}$ to be any state of $L_{12}$ such that $\left(\kappa^{\perp}|\bar{\kappa}\right)_{L_{12}}\neq0$: thus, the composite state $\left|\beta\right)_{C_1C_2}\coloneqq\left|\bar{\kappa}\right)_{L_{12}}\otimes\left|\rho_1\right)_{C_1}\otimes\left|\rho_2\right)_{C_2}$---where $\left|\rho_1\right)_{C_1}$ and $\left|\rho_2\right)_{C_2}$ are arbitrary non-null states---satisfies
\begin{equation}\label{eq:pairing}
\scalebox{0.9}{\input{Proof-latent-classical/pairing.tikz}}
\!\neq\ 
0
.
\end{equation}
Finally, we recall that the real-linear span of any finite-dimensional locally tomographic theory, including standard classical theory, contains a \emph{generalised Choi vector} $\left|\gamma_{\text{Ch}}\right)_{SS}\coloneqq\sum_i|ii)_{SS} \in\textbf{St}_{\mathbb{R}}(SS)$ and \emph{a generalised Choi covector} $\left(g_{\text{Ch}}\right|_{SS}\coloneqq\sum_j(jj|_{SS}\in\textbf{Eff}_{\mathbb{R}}(SS)$ for every system $S$, such that
\begin{equation}\label{eq:choi}
\scalebox{0.9}{\input{Proof-latent-classical/choi.tikz}}
\ .
\end{equation}
Note that neither of these is necessarily physical as they only need to live in the vector space spanned by the states/effects. In particular, even if they are in the cone of states/effects they are not necessarily normalised. We are now in position to prove the thesis by contradiction. Suppose that there exists an ontological model $\tilde{\xi}$ for the above-defined class of LCTs. Accordingly, by diagram preservation and linearity, from relation~\eqref{eq:jellyfish} one obtains
\begin{equation*}
\scalebox{0.9}{\input{Proof-latent-classical/jellyfish-ont.tikz}}\ ,
\end{equation*}
which, by exploiting relation~\eqref{eq:choi}, in turn implies
\begin{equation}\label{eq:pairing-choi}
\scalebox{0.8}{\input{Proof-latent-classical/pairing-choi.tikz}}
\ =\ 
\scalebox{0.8}{\input{Proof-latent-classical/pairing-choi-2.tikz}}
\ =\ 
\scalebox{0.8}{\input{Proof-latent-classical/null-choi.tikz}}
\!=\ 0.
\end{equation}
On the other hand, by relation~\eqref{eq:pairing} and probability preservation, one also has
\begin{equation*}
\scalebox{0.9}{\input{Proof-latent-classical/pairing-ont.tikz}}
\!\neq\ 0,
\end{equation*}
which is in contradiction with~\eqref{eq:pairing-choi}.
\endproof
Some comments are here in order. The above proof is based on more general techniques and results which will be presented in a forthcoming work~\cite{Erba2025genuine}. In particular, by generalising the proof technique exploited here, Theorem~\ref{thm:counterexample} can be shown to hold even in the case where the image of the ontological model is taken to be an \emph{infinite-dimensional classical theory} (provided that local tomography is satisfied)~\cite{Erba2025genuine}. This means that the result is not an artifact of considering ontological models into a finite-dimensional classical theory. In fact, the LCTs that we defined in Theorem~\ref{thm:counterexample} cannot be embedded in any locally tomographic theory~\cite{Erba2025genuine}. Remarkably, the example illustrated here is the first one (to our knowledge) of a theory exhibiting such a strong violation of local tomography, in the sense that it allows for the existence of a non-null bipartite effect annihilating every product state (as per Eq.~\eqref{eq:annihilating-effect}). By construction, the latter feature is common to LQTs~\cite{erba2024compositionrulequantumsystems} as well. 
\section{Conclusion}
\begin{figure}[htbp]
\centering
\includegraphics[width=1.0\textwidth]{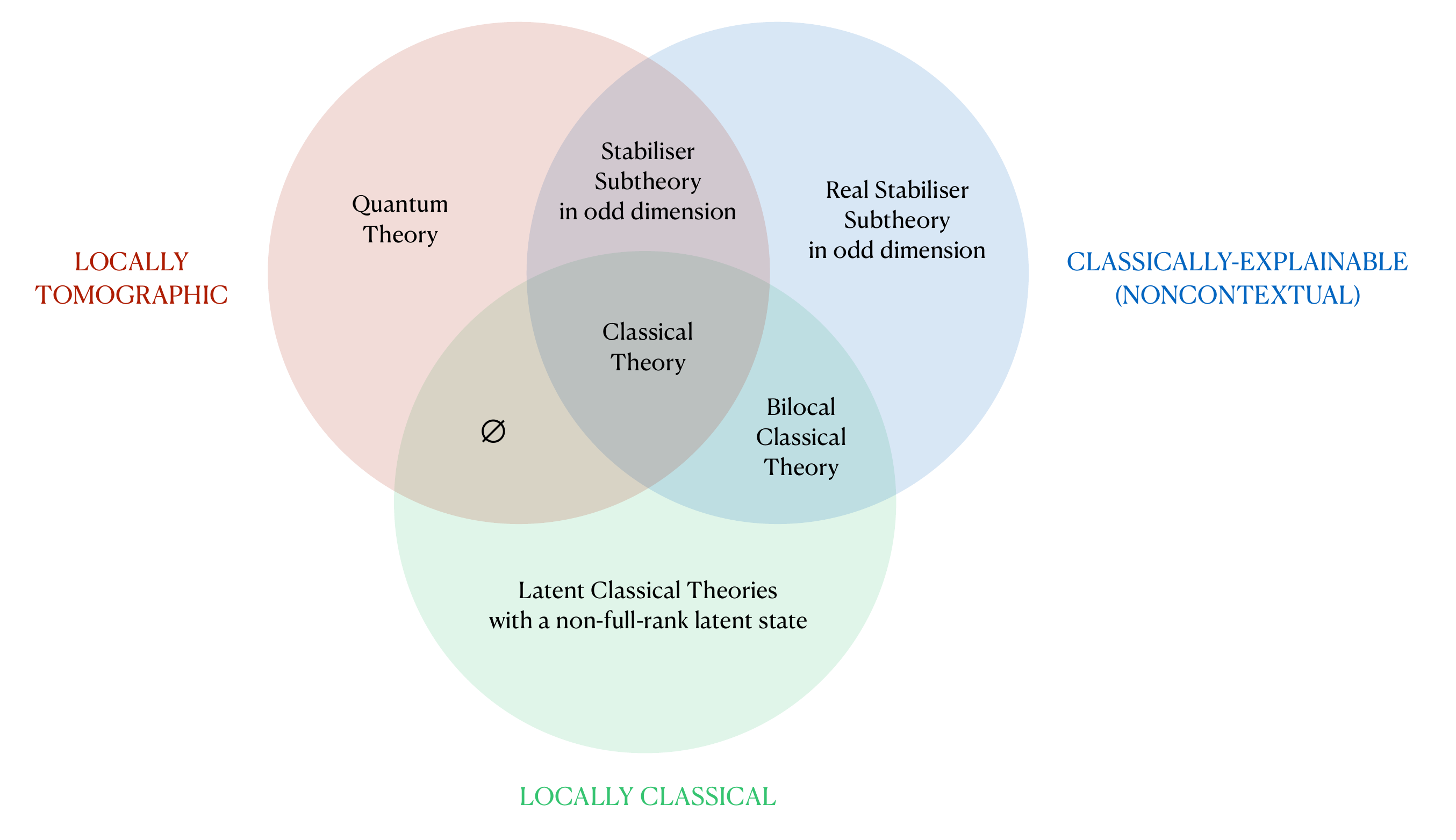}
\caption{
	Venn diagram incorporating our results. One example of a theory is included for each region, with the exception of a single region that is empty (since the only theory that is both locally tomographic and locally-classical is CT, which is, however, also classically-explainable). Our main results, Theorem~\ref{thm:BCT-NCOM} and Theorem~\ref{thm:counterexample}, allow us to populate the two regions of locally-classical theories which are non-locally tomographic.
}
\label{fig:venn-diagram}
\end{figure}
The existence of an ontological model for BCT refutes the conjecture of Ref.~\cite{d2020classicality} and shows that in general there is no tension between failures of local tomography and classical explainability (i.e., generalised noncontextuality). Indeed, the model constructed for BCT provides a way to locally account for non-locally tomographic degrees of freedom by locally adding extra ontic states. This matches with the results of Ref.~\cite{centeno2024twirled}, and can in fact be viewed as the completion of the argument regarding holism from that work from a prepare-and-measure scenario to a fully compositional theory. Moreover, it embodies a concrete example that the assumption of local tomography used in the structure theorem in Ref.~\cite{schmid2024structuretheorem} is indeed a genuine limitation of the theorem.

Our proof technique for demonstrating the validity of the ontological model for BCT (see Appendices) provides a concrete methodology to define fully-fledged ontological models for other theories existing in the literature. The explicit model constructed here is also likely to provide a useful tool for further research into locally-classical theories. On the one hand, it provides a convenient diagrammatic representation of BCT which will be useful for better understanding the features of this theory. On the other hand, it provides a convenient route towards generalising this theory, both to further locally-classical theories, but also to general methods for modifying the composition rule of a given OPT. In particular, it opens the door to understanding the relationship between BCT and the \emph{twirled} and \emph{swirled} \emph{worlds} introduced in Refs.~\cite{centeno2024twirled,ying2025quantumtheoryneedscomplex}.

We found that not every locally-classical theory is classically-explainable, as shown by the latent classical theories that we defined. Note that, for latent classical unipartite systems, every two-stage experiment where a state is prepared and a measurement is chosen out of a set of possible measurement settings (called a \emph{prepare-and-measure scenario}) is clearly classically-explainable. In fact, latent classical unipartite systems satisfy a stronger property: since they are strictly classical (therefore simplicial), they are also \emph{simplex-embeddable}~\cite{SchmidGPT} by construction---and then, in particular, also \emph{Kochen--Specker noncontextual}~\cite{KS,RevModPhys.94.045007}. Consequently, the failure to admit of an ontological model~\cite{schmid2024structuretheorem} is here due to how latent classical systems compose, essentially exhibiting a strong violation of local tomography---one that cannot be explained away due to the impossibility of embedding the theory into any locally tomographic one~\cite{Erba2025genuine}. Note that this notion of classical explainability is sensitive to this kind of failure of nonclassicality, whereas the traditional Kochen--Specker~\cite{KS,RevModPhys.94.045007} notion of noncontextuality is not, since the latter applies to single systems, whereas the generalised notion constrains all parts of a theory, including the theory's composition rule. Hence, our results, that are summarised in Figure~\ref{fig:venn-diagram}, demonstrate that there is no straightforward relationship between theories being locally-classical and being classically-explainable. It would therefore be an interesting direction to find some physical principle which singles out which locally-classical theories are classically-explainable. This also underscores the fact that considering how systems compose is an important part of understanding what it means to be classical, just as it is an important part of defining and understanding any fully-fledged theory.

\phantomsection
\addcontentsline{toc}{section}{Acknowledgments}
\section*{Acknowledgements}
The authors are grateful to Paweł Mazurek for his help in refining the title of the manuscript. SS, ME, DS, and JHS were supported by the
National Science Centre, Poland (Opus project, Categorical Foundations of the Non-Classicality of Nature,
project no.~2021/41/B/ST2/03149).
JHS conducted part of this research while visiting the Okinawa Institute
of Science and Technology (OIST) through the Theoretical Sciences Visiting Program (TSVP). DS was supported by Perimeter Institute for Theoretical Physics. Research at Perimeter Institute is supported in part by the Government of Canada through the Department of Innovation, Science and Economic Development and by the Province of Ontario through the Ministry of Colleges and Universities.

\bibliographystyle{quantum}
\phantomsection
\addcontentsline{toc}{section}{References}
\bibliography{context}

\phantomsection
\addcontentsline{toc}{section}{Appendix}
\addtocontents{toc}{\protect\setcounter{tocdepth}{0}}
\appendix

\section{Further preliminaries}\label{App:Prelim}

In this appendix we provide sufficient technical details on the background material such that this paper can be understood. However, we also provide references for more in depth and pedagogical introductions to these topics. 

\subsection{Operational probabilistic theories}\label{app:OPTs-introduction}
In this part of the appendix, we try to elaborate further on Section~\ref{Sec:OPTs} and briefly provide a more detailed review of the framework of operational probabilistic theories (OPTs), addressing additional preliminaries of the OPT framework that help clarify our main result\footnote{Or, more precisely, the framework of operational probabilistic theories which satisfy \emph{process tomography}, i.e., those OPTs that have been quotiented~\cite{soltani2025noncontextualontologicalmodelsoperational} relative to the operational equivalence relation induced by~\eqref{eqQa} below.}. For a more comprehensive review we refer readers to~\cite{d2017quantum} and the forthcoming work~\cite{erba2025categorical}. As mentioned in the main text, the primitives in the OPT framework are systems, transformations, probabilities, and \textit{instruments}.

An instrument specifies the set of possible transformations that may occur in a given probabilistic process.
whcich we denote it as $T^{A\to B}_{X}$, where $X$ represents the outcome space, i.e., the set of possible outcomes that may occur in the process. Each instrument $T^{A\to B}_{X}$ is an indexed family of transformations from $A$ to $B$, namely, 
\begin{equation}
	\label{tests}
	T^{A\to B}_{X}=[ t_{x} ]_{x\in X}; \;\;\;\forall x\in X,\;\;t_{x}\in
\textbf{Transf}(A\rightarrow B).
\end{equation} 
An instrument with a single outcome will be called a \textit{singleton
instrument}, and will have an outcome set denoted by $\star :=\{ \ast\}$.

We denote the sequential composition of two arbitrary transformations, $t_1\in\textbf{Transf}(A\to B)$ and $t_2\in\textbf{Transf}(B\to C)$, for all systems $A,B,C\in\textbf{Sys}(\Theta)$ as $t_2\circ t_1$, which is depicted as
\begin{equation}
\input{Diagramss/37_seqcomp1.tikz}=\input{Diagramss/37_seqcomp.tikz}\; .
\end{equation}
To allow the OPT framework to include the possibility of doing nothing to systems, we introduce an identity transformation, $I_{S}$, for each system $S\in\textbf{Sys}(\Theta)$, which we depict as
\begin{equation}
\input{Diagramss/38_id.tikz}\;.
\end{equation} 
Furthermore, since an identity transformation represents doing nothing to systems, we demand for an identity transformation to satisfy
\begin{equation}
I_{B}\circ t_{x}=t_{x}=t_{x}\circ I_{A},\;\;\;\forall t\in 
\textbf{Transf}(A\to B).
\end{equation}  
For all systems $A$ and $B$, in the OPT framework one can combine single systems to form a composite system $AB$, which we depict as
\begin{equation}
\input{Diagramss/38_id1.tikz}\;,
\end{equation}
where, when we use the trivial system $I \in \textbf{Sys}(\Theta)$, it satisfies $IA = A = AI$ for all $A \in \textbf{Sys}(\Theta)$.

Having composite systems allows one to compose two arbitrary transformations, $t_{1}\in\textbf{Transf}(A\to B)$ and $t_{2}\in\textbf{Transf}(C\to D)$, for all systems $A, B, C, D \in \textbf{Sys}(\Theta)$, in a parallel way, which we denote as $t_1\boxtimes t_2$, and depict as
\begin{equation}
\input{Diagramss/39_parcomp1.tikz}=\input{Diagramss/39_parcomp2.tikz}\;.
\end{equation} 
This has to be compatible with sequential composition, in the sense that 
  \[(t_2\circ t_1)\boxtimes (t_4\circ t_3) = (t_2\boxtimes t_4)\circ(t_1\boxtimes t_3).\] 
 
To describe the exchange of systems, an OPT requires the notion of swapping, that is, a family $\sigma$ of invertible transformations such that, for any two systems $A$ and $B$, their exchange is represented by $\sigma_{A,B}$, depicted as
\begin{equation}
\input{Diagramss/41_swap1.tikz}=\input{Diagramss/41_swap.tikz}\; ,
\end{equation}
and its inverse, denoted by $\sigma^{-1}_{A,B}$, is depicted as
\begin{equation}
\input{Diagramss/41_swap3.tikz}=\input{Diagramss/41_swap2.tikz}\; .
\end{equation}
Swapping is required to satisfy the sliding property, namely,
\begin{equation}
\input{Diagramss/41_swap4.tikz}=\input{Diagramss/41_swap5.tikz}\;,
\end{equation}
for all $A, B, C, D \in\textbf{Sys}(\Theta)$, and all $t_{1}\in \textbf{Transf}(A\to B)$, $t_{2}\in \textbf{Transf}(C\to D)$. In the case where $\sigma^{-1}_{A,B}= \sigma_{B,A}$, the OPT is said to be symmetric, where we represent $\sigma_{A,B}$ as
\begin{equation}
\input{Diagramss/41_swap1.tikz}=\input{Diagramss/41_swap6.tikz}\;.
\end{equation} 
Note that the compositional structure of an OPT, as described here, corresponds to the structure of a braided monoidal category.

Given a state $\rho$ for system $A$, followed by a transformation $t$ from $A$ to $B$, and then an effect $a$, an OPT $\Theta$ assigns a probability $p \in [0,1]$ to 
\begin{equation}
\label{probability}
\input{Diagramss/43_joinp.tikz}=\input{Diagramss/62_prob2.tikz}=p\;.
\end{equation}
This probability $p$ is the chance that $\rho$ occurs among all other states in its corresponding state instrument, $t$ among all transformations in its instrument, and $a$ among all outcomes in its corresponding effect instrument. The assignment thus defines a probability distribution over all transformations in the respective instruments. We define parallel composition of scalars as
\begin{equation}
p\boxtimes q:=pq,\;\;\;\forall p,q\in \textbf{Transf}(I\to I),
\end{equation}
which is simply the multiplication of real numbers. We can also define \textit{multiplication by scalars}, that is, for every $p\in\textbf{Transf}(I\to I)$ and every $t\in\textbf{Transf}(A\to B)$, one has
\begin{equation}
\label{scmul}
p\cdot\input{Diagramss/36_event_t.tikz}:=\input{Diagramss/42_scaler.tikz}=\input{Diagramss/42_scaler1.tikz}\;\;.
\end{equation} 
an OPT $\Theta$ is required to be \emph{tomographic} for transformations. That is, one can use the states and effects within an OPT to generate statistics to distinguish them or to identify their equivalence. In particular, for every $t_{1}, t_{2} \in \textbf{Transf}(A \to B)$ one has $t_{1} = t_{2}$ if and only if they yield the same probabilities, i.e.,
\begin{equation}
\label{eqQa}
\input{Diagramss/44_eq.tikz}=\input{Diagramss/44_eq1.tikz}
\end{equation}
for every $E \in \textbf{Sys}(\Theta)$, every state $\rho \in \textbf{Transf}(I \to AE)$, and every effect $a \in \textbf{Transf}(BE \to I)$. Note that in the full framework of OPTs this is in fact only required to be an equivalence relation rather than an equality, however, one can always quotient with respect to this relation and obtain an OPT of the sort we have defined here. 

Note that in an OPT, as can be seen from~\eqref{eqQa} that $\textbf{St}(A) := \textbf{Transf}(I\to A)$ can be identified with a set of functionals from $\textbf{Eff}(A):=\textbf{Transf}(A \to I)$ to the real interval $\left[0, 1\right]$, and vice versa. Moreover, the set of states separates the set of effects, and the set of effects separates the set of states. This means that for any pair of distinct states $\vert\rho)_{A}, \vert\sigma)_{A} \in \textbf{St}(A)$, there exists an effect $(a\vert_{A} \in \textbf{Eff}(A)$ such that $(a\vert \rho)_{A} \neq (a\vert \sigma)_{A}$ (and similarly for any pair of distinct effects), where one can show~\cite{chiribella2010probabilistic,Chiribella_2014,chiribella2016quantum,d2017quantum} that these two properties imply a linear structure, as described in the main text. Similarly, transformations from $A$ to $B$ can be regarded as functionals that assign a probability to each state–effect pair. As such, they also form a vector space, whose linear structure is discussed in the main text.

Also, note that, as a consequence of~(\ref{eqQa}), any transformation $t \in \textbf{Transf}(A\to B )$ naturally gives rise to a linear map $\hat{t}_{E}$ from $\textbf{St}_\mathbb{R}(AE)$ to $\textbf{St}_\mathbb{R}(BE)$ for every system $E$. Thus, each transformation $t$ is fully specified by the collection of such linear maps $\lbrace \hat{t}_{E} \rbrace_{E \in \textbf{Sys}(\Theta)}$. In particular, two transformations $t_1$ and $t_2$ from $A$ to $B$ are equal if and only if
\begin{equation}
\label{eq:local-tomography-no-ancilla}
t_1 = t_2 \iff \hat{t}_{1,E} = \hat{t}_{2,E},\quad \forall E \in \textbf{Sys}(\Theta).
\end{equation}
In general, the full infinite family $\lbrace \hat{t}_{E}\rbrace_{ E \in \text{Sys}(\Theta)}$ over all systems $E$ may be needed to completely determine a transformation. However, in some cases, this is not required. For instance, in ``locally tomographic'' theories, each transformation is uniquely determined by its action on the system alone~\cite{schmid2024structuretheorem}, meaning the map with trivial ancilla $E$ suffices.

In a locally tomographic OPT, the complete specification of any state is achievable using only the statistics of local measurements. That is, two states are identical if they yield the same statistics for all local measurements. An OPT $\Theta$ is locally tomographic if and only if the dimension condition $D_{AB} = D_A D_B$ holds for all systems $A, B \in \textbf{Sys}(\Theta)$~\cite{d2017quantum}. Under this condition, for each composite system $AB$, the real vector space $\textbf{St}_{\mathbb{R}}(AB)$ is spanned by product states. Moreover, to fully determine a transformation $t \in \textbf{Transf}(A \to B)$ it suffices to examine its action on the system alone, using only the trivial system as the ancilla. 

The notion of local tomography can be generalised to \textit{$n$-local tomography}~\cite{d2020classicality, hardy2012limited}, wherein a theory exhibits $n$-local tomography if the state of a composite system can be completely determined using measurement statistics from subsystems of up to $n$ components—that is, from $1$-component, $2$-component, ..., up to $n$-component measurements. For the special case $n=2$, the theory is referred to as a \textit{bilocally tomographic} theory. In such a theory, both local and bipartite measurements—as well as conical combinations of their parallel compositions—are required to fully characterise a multipartite state. An example of such theory is BCT reviewed in Sec.~\ref{Sec:BCT}.

A set of deterministic (i.e., normalised) states $\lbrace\vert\rho_{i} )\rbrace_{i=1}^{n}$ for a system $A$ is said to be jointly perfectly discriminable if there exists an observation instrument $\lbrace ( a_{i}\vert _{\text{A}} \rbrace _{i=1}^{n}$ such that
\begin{equation}
\label{eq:perfectly-discriminable}
\left(a_{i}\vert \rho_{i^{'}}\right)=\delta_{i, i^{'}},\;\;\;\forall i, i^{'}\in\lbrace1,2,\ldots,n\rbrace.
\end{equation} 

An agent may choose to conflate outcomes within any subset of a instrument's outcome space, that is, ignore distinctions between transformations in that subset. To formalise this in the OPT framework, we define coarse-graining for an instrument $T^{A\to B}_{X} = [t_{x}]_{x\in X}$ and any subset $Y \subseteq X$ as
\begin{equation}
t_{Y}:=\sum_{y\in Y} t_y,
\end{equation}
where both sequential and parallel composition distribute over summation.

Every deterministic transformation (see~\eqref{eq:normalisation}) 
$t \in \textbf{Transf}(A \to B)$ defines a singleton instrument such that $T_{\star}=\lbrace t \rbrace$. which implies that that the transformation occurs with certainty, having a marginal probability of 1. Consequently, since singleton instruments correspond to deterministic transformations, for any instrument $T_{X}^{A\to B}=\lbrace t_x\rbrace_{x\in X}\in \textbf{Instr}(A\to B)$, the full coarse-graining, i.e., $t_X=\sum_{x\in X} t_x$, is a deterministic transformation. Deterministic transformations are referred to as \textit{channels}, while deterministic extremal states are known as \textit{pure states}. 

A transformation $\mathscr{R}\in\textbf{Transf}(A\to B)$ is called \textit{reversible}, if there exists a transformation $\mathscr{R}^{-1}\in\textbf{Sys}(A\to B)$ such that $\mathscr{R}^{-1}\mathscr{R}=\mathcal{I}_A$ and $\mathscr{R}\mathscr{R}^{-1}=\mathcal{I}_B$.

To account for zero probabilities within the framework, we assume $0 \in \textbf{Transf}(I \to I)$. From scalar multiplication~\eqref{scmul}, it follows that for any $t \in \textbf{Transf}(A \to B)$ and any systems $A, B \in \textbf{Sys}(\Theta)$, one can define a \textit{null} transformation $\varepsilon_{A \to B} := 0 \cdot t \in \textbf{Transf}(A \to B)$. This null transformation satisfies the condition that for all $E \in \textbf{Sys}(\Theta)$, $\vert \rho)_{AE} \in \textbf{St}(AE)$, and $(a\vert_{BE} \in \textbf{Eff}(BE)$, one has
\begin{equation}
\input{Diagramss/64_epsilon.tikz}=0\;.
\end{equation}
In the OPT framework, one may also consider \textit{conditional instruments}~\cite{chiribella2010probabilistic}, where the selection of an instrument depends on the outcome of a previously performed instrument. That is, based on the result of one instrument—i.e., an element of the outcome space—a subsequent instrument, conditioned on that outcome, is applied in sequence. This notion enables the definition of a stricter form of causality, known as \textit{strong causality}. An OPT is said to be strongly causal if every conditional instrument is itself an instrument of the theory. In other words, implementing a conditional instrument must result in a valid instrument within the theory. The reason this property is referred to as strong causality is that, in any OPT satisfying it, the ability to implement arbitrary conditional instruments implies the uniqueness of the deterministic effect for every system—i.e., standard causality. Hence, causality is a necessary condition for strong causality. However, the converse does not necessarily hold. 
In this paper, we do not attempt to formalise the notion of strong causality, as its technical details are not directly relevant to our purposes. 
\begin{remark}
A key feature of the OPT framework is that the primitive objects are instruments. In this formulation, the category of transformations is induced by the category of instruments, rather than the other way around. However, for our purposes, we are primarily interested in transformations rather than instruments, and we do not investigate the full implications of treating instruments as primitive. This is because our goal is to construct, or show the nonexistence of, an ontological model for specific OPTs. Ref.~\cite{soltani2025noncontextualontologicalmodelsoperational} showed that, for strongly causal OPTs, including the OPTs of interest in this paper, the construction of an ontological model can be formulated at the level of transformations without the need to include instrument structure. Our analysis of OPTs will therefore primarily focus on transformations.
\end{remark} 

Finally, it is intended that the concrete models of OPTs presented in the following all satisfy the property of \emph{existence and uniqueness of system-decomposition}~\cite[\S6, Def.~14]{Rolino_2025}, that is: every system can be \emph{uniquely} decomposed as a parallel composition of non-trivial elementary systems. This property, which was already exploited for the construction of \emph{latent theories} in Ref.~\cite{erba2024compositionrulequantumsystems}, allows for the concrete definition of parallel composition rules by specifying the least number of equalities between different decompositions of arbitrary systems. More specifically, one may start by defining all of the elementary systems for a theory, so that the composite ones are defined by strings on the alphabet of the elementary ones. For instance, finite-dimensional quantum theory may be defined by setting qudits $Q_d$, for every integer $d\geq 2$, to be the (non-trivial) elementary quantum systems; thus, instead of strict equalities, one generally has just isomorphisms of the kind
\begin{equation*}
Q_l\otimes Q_{m}\cong Q_m\otimes Q_{l},\quad Q_l\otimes Q_{m}\cong Q_{lm},\quad 
Q_l\otimes Q_{mn}\cong Q_m\otimes Q_{ln}
.
\end{equation*}
Accordingly, it will be sufficient to specify the elementary systems of an OPT together with a parallel composition rule for them: arbitrary systems will be then given by the monoid that is freely generated by 
the elementary systems. From a category-theoretic viewpoint, an OPT satisfying unique system-decomposition is a \emph{coloured PROP}~\cite{maclane1965categorical,carette2022propificationscalablecomonad}.

\subsection{Bilocal classical theory}\label{App:BCT}

The action of a local pure effect on a bipartite pure state, as well as the action of a local pure state on a bipartite pure effect, is given by
\begin{equation}
\input{BCT-diagrams/13_BCT3.tikz}=\delta_{i,i'}\input{BCT-diagrams/13_BCT4.tikz}\; ,
\end{equation}
\begin{equation}
\input{BCT-diagrams/13_BCT5.tikz}=\frac{1}{2}\delta_{i,i'}\input{BCT-diagrams/13_BCT6.tikz}\; .
\end{equation}

Sequential composition of atomic transformations is also an atomic transformation, given by
\begin{equation}
\label{seq-comp-atomic}
\input{BCT-diagrams/13_BCT23.tikz} = \input{BCT-diagrams/13_BCT24.tikz} \;,
\end{equation}
for all systems $n$, $q$, and $m$, all $1 \leq i_0 \leq n$, $1 \leq l \leq q$, $1 \leq i'_0 \leq q$, $1 \leq l' \leq m$, all $\tau, \tau' \in \lbrace 0,1\rbrace$, and for all $\lambda, \lambda' \in [0,1]$, where $\tilde{\tau} := \tau \oplus \tau'$ and $\tilde{\lambda} := \delta_{l,i'_0} \lambda \lambda'$. In the case of $m=1$, the corresponding effect on the right-hand side will also be an atomic effect. Using~\eqref{atomic-comp}, It can also be shown that the sequential composition of atomic transformations on composite systems results in another atomic transformation. However, in the case of parallel composition of atomic transformations, the resulting transformation is no longer atomic. The same holds for the parallel composition of atomic effects, where the resulting effect is no longer atomic.

Swapping $\sigma_{n,m}$ between two arbitrary systems $n$ and $m$ is defined as a transformation in BCT that satisfies
\begin{equation}
\label{swap}
\input{BCT-diagrams/6_SwapOn3ParS.tikz}=\input{BCT-diagrams/6_SwapOn3ParS1.tikz}\;,
\end{equation}
for every $1\leq i \leq n$, $1\leq j \leq m$, $1\leq k\leq n_E$, and every $s,t\in\lbrace 0,1\rbrace$. Swapping in BCT satisfies $\sigma_{n,m}^{-1} = \sigma_{m,n}$ for every system $m$ and $n$. Hence, BCT is symmetric.

\subsubsection{Proof of Proposition~\ref{uniqueness}}
\label{app:uniqueness-proof}
Consider two possible conical combinations of normalised atomic transformations for an arbitrary transformations $t$:
\begin{equation}
	\label{two-conical-combination}
	\input{BCT-diagrams/15_ea.tikz}=\sum_{i_0,l,\tau}C(i_0,l,\tau)\input{BCT-diagrams/15_ea2.tikz}=\sum_{i_0,l,\tau}C'(i_0,l,\tau)\input{BCT-diagrams/15_ea2.tikz}\;.
\end{equation}
This implies that for the null transformation from system $n$ to $m$, one has
\begin{equation}
	\label{null}
	\input{BCT-diagrams/13_BCT25.tikz}=\sum_{i_0,l,\tau}\left( C(i_0,l,\tau)-C'(i_0,l,\tau) \right)\input{BCT-diagrams/15_ea2.tikz}=\sum_{i_0,l,\tau} a(i_0,l,\tau)\input{BCT-diagrams/15_ea2.tikz}\;,
\end{equation}
where $a(i_0,l,\tau):= C(i_0,l,\tau)-C'(i_0,l,\tau)$. Now, consider an arbitrary bipartite state $ \rho $ for a bipartite system $ n \boxtimes E $, and an arbitrary bipartite effect $ a $ for a bipartite system $ m \boxtimes E $. Since $\varepsilon$ is a null transformation from system $n$ to $m$, for any such $\rho$, effect $a$, and any ancillary system $E$, one has
\begin{equation}
	\input{BCT-diagrams/13_BCT26.tikz}=0\;.
\end{equation}
If we choose
\begin{equation}
	\input{BCT-diagrams/13_BCT27.tikz}=\input{BCT-diagrams/13_BCT28.tikz}\;,
\end{equation}
for some $1 \leq i'_0 \leq n$, $1 \leq j' \leq n_E$ (where $n_E$ is the dimension of the ancillary system $E$), and some $s' \in \lbrace 0,1 \rbrace$, and also choose
\begin{equation}
	\input{BCT-diagrams/13_BCT29.tikz}=\input{BCT-diagrams/13_BCT30.tikz}\;,
\end{equation}
for some $1 \leq l' \leq m$, and some $\tau' \in \lbrace 0,1 \rbrace$,  
then using~\eqref{null} implies
\begin{align}
	\input{BCT-diagrams/13_BCT26.tikz}&=\sum_{i_0,l,\tau}a(i_0,l,\tau)\input{BCT-diagrams/13_BCT31.tikz}\nonumber\\
	&=\sum_{i_0,l,\tau}a(i_0,l,\tau)\delta_{i_0,i'_0}\input{BCT-diagrams/13_BCT32.tikz}\nonumber\\
	&=\sum_{i_0,l,\tau}a(i_0,l,\tau)\delta_{i_0,i'_0}\delta_{l,l'}\delta_{\tau,\tau'}=a(i'_0,l',\tau')=0\;,
\end{align}
where by appropriate choice of $\rho$ and $a$, we can show that all the coefficients $a(i_0,l,\tau)$, for every $i_0$, $l$, and $\tau$, are zero as well. This concludes the proof as it implies that $C(i_0,l,\tau) = C'(i_0,l,\tau)$. 

\vspace{3mm}

\subsection{Maps between different OPTs}\label{app:GNC}
In this section, we introduce the general notion of a map between transformations of two operational probabilistic theories (OPTs). For the purposes of this paper, we adapt the framework of~\cite{soltani2025noncontextualontologicalmodelsoperational}—without loss of generality—to our setting.

\begin{definition}[Maps between two OPTs] A map between two OPTs $\input{Diagramss/51_mapopt.tikz}$, takes every system $\text{A} \in \textbf{Sys}(\Theta)$ to a system $\phi(\text{A}) \in \textbf{Sys}(\Phi)$ and maps every transformation $t$ of $\Theta $ to a transformation $\phi{(t)}$ of $\Phi$, depicted as
\begin{equation}
	\phi::\input{Diagramss/36_test.tikz}\mapsto \input{Diagramss/51_mapopt2.tikz}\;,
\end{equation}
where for the trivial system, one has $\phi(I)=I$.
\end{definition}
Note that this definition implies that a map $\phi$, takes every state in $\Theta$ to a state in $\Phi$, every effect in $\Theta$ to an effect in $\Phi$, and every scalar in $\Phi$ to a scalar in $\Theta$. This definition also implies that every instrument $\Theta$ is mapped to an instrument in $\Phi$, where each transformation of the instrument in $\Phi$, is simply given by the action of the map on the corresponding transformation in $\Theta$.

A map $\phi$ is said to be \textit{linear} if for every $r_{x}\in\mathbb{R}$, every system $A$ and $B\in\textbf{Sys}(\Theta)$, and every transformation $t_{x}\in\textbf{Transf}(A\rightarrow B)$, such that $\sum_{x}r_{x}t_{x}\in\textbf{Transf}(A\rightarrow B )$, one has
\begin{equation}
\phi{\left(\sum_{x}r_{x}t_{x}\right)}=\sum_{x} r_{x}\phi{\left(t_{x}\right)}.
\end{equation}

A map is said to be \textit{diagram-preserving} $\input{Diagramss/52_dp.tikz}$ if composing transformations before or after applying the map remains equivalent. For example,
\begin{align}
\label{dpexample}
\input{Diagramss/52_dp2.tikz}=\input{Diagramss/52_dp3.tikz}\;\;.
\end{align}

Diagram-preserving maps also preserve the identities and swaps, that is
\begin{equation}
\label{identity-swap}\input{Diagramss/52_dp8.tikz}= \begin{tikzpicture}
	\begin{pgfonlayer}{nodelayer}
		\node [style=none] (0) at (-2, 0) {};
		\node [style=none] (1) at (2, 0) {};
		\node [style=none] (2) at (0, 0.5) {\scalebox{0.82}{$\phi_{\text{\tiny DP}}(\scalebox{0.8}{$\text{A}$})$}
};
	\end{pgfonlayer}
	\begin{pgfonlayer}{edgelayer}
		\draw [qWire] (0.center) to (1.center);
	\end{pgfonlayer}
\end{tikzpicture}
\;\; ,\;\; \input{Diagramss/52_dp10.tikz}=\input{Diagramss/52_dp11.tikz}.
\end{equation}

A map is said to be \textit{probability-preserving} $\input{Diagramss/59_pp.tikz}$, if for every scalar $p\in [0,1]$, one has
\begin{equation}
\input{Diagramss/59_pp5.tikz}=\input{Diagramss/59_pp6.tikz}\;.
\end{equation}
A map is said to be \textit{determinacy preserving} if it maps every deterministic transformation to a deterministic transformation in the codomain OPT. For the unique deterministic effect in causal OPTs, this implies
\begin{equation}
\input{BCT-diagrams/sina.tikz}=\begin{tikzpicture}
	\begin{pgfonlayer}{nodelayer}
		\node [style=infupground] (0) at (0.5, 0) {};
		\node [style=none] (1) at (-1.25, 0) {};
		\node [style=none] (2) at (-0.75, 0.5) {\scriptsize$\phi (A)$};
	\end{pgfonlayer}
	\begin{pgfonlayer}{edgelayer}
		\draw [qWire] (0) to (1.center);
	\end{pgfonlayer}
\end{tikzpicture}
\;.
\end{equation}

\begin{remark}\label{rem:ont-models-structure} 
Note that our definition of a general map between OPTs acts only on transformations. Although one could instead define such maps, including ontological models, at the level of instruments, Ref.~\cite{soltani2025noncontextualontologicalmodelsoperational} shows that this is unnecessary for strongly causal quotiented of OPTs, including BCT and latent classical theories. In these cases, noncontextual ontological models are fully determined at the level of transformations, as in Def.~\ref{def:ONT}.

In general, however, for OPTs that are not strongly causal, this definition may be inadequate, since the instrument structure may still induce a contextual representation of ontological models even at the level of quotiented OPTs. Consequently, the mere existence of ontological models defined only at the level of transformations, as in Def.~\ref{def:ONT}, would not necessarily imply classical explainability (see Ref.~\cite{soltani2025noncontextualontologicalmodelsoperational} for more details).
\end{remark}
\begin{remark}\label{rem:ont-models-structure2}
Once instrument structure is ignored, a quotiented OPT can be regarded as a particular kind of strict symmetric monoidal category enriched in finite-dimensional real vector spaces, together with a specification of which hom-vectors represent physical transformations and which are merely linear combinations of them. Somewhat more specifically, the physical transformations in each hom-vector space can be characterized by the intersection of an appropriate cone with a set of linear inequality constraints. From this perspective, a noncontextual ontological model can be viewed as a linear strict monoidal functor $\input{Diagramss/47_ontmodelunq.tikz}$ that maps physical transformations to physical transformations. More explicitly, diagram preservation supplies the strict monoidal functorial structure, linearity ensures linearity on the hom-vector spaces, and determinism preservation ensures that deterministic transformations are mapped to deterministic transformations. Probability preservation follows from these assumptions, as noted in the main text. 
\end{remark}

\section{Relating single- and multi-system definitions}\label{App:MultiVsSingle}

In the main text we define BCT and its ontological model using multi-system transformations, states, and effects. This was a choice for simplicity of presentation, and here we will show how to relate the single-system and multi-system definitions via explicit merging and splitting maps. 

\subsection{In bilocal classical theory}

To do this, we first introduce a invertible merging transformations for each pair of systems, denoted by $\nu$ where the inverse $\nu^{-1}$ is then the splitting transformation. These satisfy
\begin{equation}
	\label{nu0}
	\input{BCT-diagrams/11_nu.tikz}=\input{BCT-diagrams/11_nu1.tikz},
\end{equation}
for every system $n_1$ and $n_2$. In particular, these map pure bipartite states to  pure states of the corresponding single system, such that
\begin{equation}
	\label{nu1}
	\input{BCT-diagrams/11_nu7.tikz}=\input{BCT-diagrams/11_nu8.tikz}\;,
\end{equation}
for all systems $n_1$, $n_2$, and for all $1\leq i\leq n_1$, $1\leq j\leq n_2$, and $s\in\lbrace 0,1\rbrace$, where we define
\begin{equation}
	Q(i, j, s) := 2n_1(i - 1) + (2j + s - 1).
\end{equation} 
In other words, $\nu$ maps each bipartite pure state with subsystems $n_1$ and $n_2$ to a pure state of a single system of size $2n_1n_2$. For each combination of $i$, $j$, and $s$, it assigns the pure state labeled by $Q(i, j, s)$, where $1 \leq Q(i, j, s) \leq 2n_1n_2$.

In order to fully specify $\nu$ we must specify its action on arbitrary multipartite pure states of any composite system in BCT, in particular we have that
\begin{equation}
	\label{nu2}
	\input{BCT-diagrams/12_mud25.tikz}=\input{BCT-diagrams/12_mud26.tikz}\;,
\end{equation}
where the empty triangle  hereon  represents $\nu$. In other words, to a pure state of a composite system with $p$ single subsystems, it assigns another pure composite state, where one has $p-1$ single subsystems, as described above. For example, for an arbitrary tripartite pure state, one has
\begin{equation}
	\input{BCT-diagrams/11_nu4.tikz}=\input{BCT-diagrams/11_nu5.tikz}\; ,
\end{equation}
where applying $\nu$ again yields the corresponding single-system pure state. More generally, for an arbitrary composite system consisting of single subsystems $n_1, n_2, n_3, \ldots, n_p$, and an arbitrary pure state $\rho$, using $\nu$ consecutively $p-1$ times as
\begin{equation}
	\input{BCT-diagrams/13_BCT18.tikz}\;,
\end{equation}
results in the corresponding single-system pure state with the dimension $2^{p-1}n_1n_2n_3\ldots n_{p-1}n_p$, which we denote it as $\mathcal{Q}_{p}\left[\rho\right]$. For example, for the cases of bipartite, tripartite and quadripartite pure states one has
\begin{align}
	\mathcal{Q}_{2}\left[ (ij)_s \right] 
	&= Q\left( i, j, s \right) , \nonumber \\
	\mathcal{Q}_{3}\left[ \big( (ij)_s \, k \big)_t \right] 
	&= Q\big( Q\left( i, j, s \right), k, t \big) ,\nonumber \\
	\mathcal{Q}_{4}\left[ \Big( \big( (ij)_s \, k \big)_t \, l \Big)_u \right] 
	&= Q\Big( Q\big( Q\left( i, j, s \right), k, t \big), l, u \Big).
\end{align}
And for an arbitrary pure multipartite state, it can be depicted as
\begin{equation}
	\label{nu-on-multipartite-state}
	\input{BCT-diagrams/12_mud9.tikz}:=\input{BCT-diagrams/13_BCT18.tikz}\; ,
\end{equation}
where $N=2^{p-1}n_1n_2n_3\ldots n_{p-1}n_p$. Similarly, for a multipartite pure effect one has
\begin{equation}
	\label{nu-on-multipartite-effect}
	\input{BCT-diagrams/12_mud42.tikz}:=\input{BCT-diagrams/12_mud43.tikz}\;,
\end{equation}
where the empty triangles represent $\nu^{-1}$. The same logic can be used for normalised atomic transformations. Consider an arbitrary normalised atomic transformation on composite systems
\begin{equation}
	\input{BCT-diagrams/13_BCT33.tikz}\; .
\end{equation}
We can associate an effect $a_l$ to the lower indices as 
\begin{equation}
	i_{1},i_{2},\cdots i_{p},s_{1},s_{2}\cdots,s_{p-1} \;\;\longrightarrow \;\; a_{l}= \left(\left(\cdots\left(\left(i_{1}i_{2}\right)_{s_{1}}i_{3}\right)_{s_{2}}\cdots i_{p-1}\right)_{s_{p-2}}i_{p}\right)_{s_{p-1}},
\end{equation}
and a state $\rho_u $ to the upper indices as
\begin{equation}
	l_{1},l_{2},\cdots,l_{q},t_{1},t_{2}\cdots t_{q-1}\;\;\longrightarrow \;\; \rho_u= \left(\left(\cdots\left(\left(l_{1}l_{2}\right)_{t_{1}}l_{3}\right)_{t_{2}}\cdots l_{p-1}\right)_{t_{p-2}}l_{p}\right)_{t_{q-1}}.
\end{equation}
Using $a_l$ and $\rho_u$, for an arbitrary normalised atomic transformation on composite systems, one has
\begin{align}
	\label{nu-on-multipartite-atomic}
	&\hspace{4.3cm}\input{BCT-diagrams/13_BCT33.tikz}\nonumber\\
	&=\input{BCT-diagrams/12_mud37.tikz}\;.
\end{align}

Since $\nu$ is a transformation in BCT, it can be expressed as a conical combination of normalised atomic transformations (note that in App.~\ref{atomiccomp} we show why Proposition~\ref{uniqueness} also holds for transformations on composite systems such as $\nu$), given by
\begin{equation}
	\label{nu-atomic}
	\input{BCT-diagrams/11_nu10.tikz}=\sum_{i,j,s,l,\tau}C\left(i,j,s,l,\tau\right)\cdot\input{BCT-diagrams/12_mud1.tikz}\; ,
\end{equation}

which, when applied to an arbitrary pure state $\left(\left(i' j' \right)_{s'} k' \right)_{t'}$, results in
\begin{equation}
\begin{split}
	\input{BCT-diagrams/11_nu9.tikz}
    &=\input{BCT-diagrams/12_mud2.tikz}=
    \\ &=\sum_{i,j,s,l,\tau}C\left(i,j,s,l,\tau\right)\delta_{i,i^{'}}\delta_{j,j^{'}}\delta_{s,s^{'}}\input{BCT-diagrams/12_mud3.tikz}\; .
\end{split}
\end{equation}
This, as a consequence of the uniqueness of the conical combination of normalised atomic transformations established in Proposition~\ref{uniqueness}, uniquely determines the conical coefficients as
\begin{equation}
	C\left(i,j,s,l,\tau\right)=\delta_{l,Q\left(i,j,s\right)}\delta_{\tau,0}.
\end{equation}
Substituting the above equation into \eqref{nu-atomic} implies
\begin{equation}
	\label{nu-atomic2}
	\input{BCT-diagrams/11_nu10.tikz}=\sum_{i,j,s}\input{BCT-diagrams/12_mud4.tikz}\; .
\end{equation}
Similarly, using the same logic, one has
\begin{equation}
	\input{BCT-diagrams/11_nu11.tikz}=\sum_{i,j,s}\input{BCT-diagrams/11_nu12.tikz}\; .
\end{equation}

\subsection{In the ontological model}

We map the splitting and merging transformations of BCT, namely $\nu$ and its inverse as introduced in~\eqref{nu0},~\eqref{nu1}, and~\eqref{nu2}, to define splittings and mergings for the ontological model, namely $\mu$ and its inverse $\mu^{-1}$. That is,
\begin{equation}
\label{eq:mu}
	\input{BCT-diagrams/12_mud15.tikz}\; ,\;\;\input{BCT-diagrams/12_mud16.tikz}\; ,
\end{equation}
where applying the ontological model on the both sides of \eqref{nu-atomic2} implies
\begin{equation}
	\label{mu-derivation}
	\input{BCT-diagrams/12_mud6.tikz}=\sum_{i,j,s}\input{BCT-diagrams/12_mud5.tikz}=\input{BCT-diagrams/4_OntMapOfStJ15.tikz}.
\end{equation}
Here, $\otimes$ is defined as the merging gate
\begin{equation}
	\input{BCT-diagrams/4_OntMapOfStJ16.tikz}\coloneq
	\input{BCT-diagrams/4_OntMapOfStJ17.tikz}\; .
\end{equation}
The second equality in \eqref{mu-derivation} can be verified when we apply both sides to the orthonormal basis of the corresponding classical space. Similarly, for $\mu^{-1}$, one has
\begin{equation}
	\input{BCT-diagrams/12_mud7.tikz}=  \input{BCT-diagrams/12_mud41.tikz}\; ,
\end{equation}
wherein $\begin{tikzpicture}
	\begin{pgfonlayer}{nodelayer}
		\node [style=dmerge] (0) at (0, 0) {};
	\end{pgfonlayer}
\end{tikzpicture}
$ is defined as the inverse of the merging gate $\otimes$
\begin{equation}
	\input{BCT-diagrams/4_OntMapOfStJ19.tikz}\coloneq\input{BCT-diagrams/4_OntMapOfStJ20.tikz}\; .
\end{equation}
Having defined $\nu$ and consequently $\mu$, we can alternatively define the ontological model only for single-system entities, namely, single-system pure states, single-system pure effects, and single-system normalised atomic transformations. This, in turn, enables us to derive their multipartite counterparts, which we will show to be consistent with the direct definitions we have given for multipartite entities in the ontological model.

Consider a pure bipartite state \eqref{nu1}. If we apply the ontological model on the both sides, we have
\begin{equation}
	\label{mu-on-bipartie}
	\input{BCT-diagrams/12_mud30.tikz}=\input{BCT-diagrams/12_mud13.tikz}=\input{BCT-diagrams/12_mud14.tikz}\;.
\end{equation}
Note that in the first equality, we use diagram preservation. This indeed holds in this particular case which we show in App.~\ref{diagram-preservation-of-mu}.   

Now, having derived $\mu$ and its inverse, from \eqref{mu-on-bipartie}, one has
\begin{equation}
	\input{BCT-diagrams/12_mud17.tikz}=\input{BCT-diagrams/12_mud18.tikz}\;.
\end{equation}
In other words, this implies that for a bipartite pure state, instead of applying the ontological model directly to the pure bipartite state, we can first apply the ontological model to the corresponding pure state of the corresponding single system, and then apply $\mu^{-1}$, which entails
\begin{equation}
	\input{BCT-diagrams/12_mud17.tikz}=\input{BCT-diagrams/12_mud19.tikz}=\input{BCT-diagrams/4_OntMapOfStJ25.tikz}\;,
\end{equation}
which is in agreement with the multi-system definition in Eq.~\eqref{ont-model-on-composite-state}.
This can be generalised to a pure multipartite state~\eqref{nu-on-multipartite-state} as 
\begin{equation}
	\label{composite-satate-as-single-states}
	\input{BCT-diagrams/12_mud20.tikz}=\input{BCT-diagrams/12_mud21.tikz}\; ,
\end{equation}
where empty triangles represent $\mu^{-1}$. When clear from context, we use empty triangles to represent both $\mu$ and its inverses. For a pure multipartite effect~\eqref{nu-on-multipartite-effect}, one has
\begin{equation}
	\label{composite-effects-as-single-effects}
	\scalebox{0.9}{\input{BCT-diagrams/12_mud22.tikz}}
    =
    \scalebox{0.9}{\input{BCT-diagrams/12_mud23.tikz}}\; .
\end{equation}
Note that to derive both~\eqref{composite-satate-as-single-states} and~\eqref{composite-effects-as-single-effects}, we assume diagram preservation for the ontological model, where in App.~\ref{diagram-preservation-of-mu} we show why it holds for both of these cases.

For multipartite normalised atomic transformations~\eqref{nu-on-multipartite-atomic}, one has
\begin{align}
	\label{comp-atomic-mu}
	&\hspace{4.9cm}\input{BCT-diagrams/3_CompAtomicTrans5.tikz}\nonumber\\
	&=\scalebox{0.85}{\input{BCT-diagrams/12_mud24.tikz}}\; .
\end{align}
Once again, to derive the equation above, diagram preservation is required, which we show why it holds for this specific case in App.~\ref{diagram-preservation-of-mu}. 
To better understand how this works, consider a normalised atomic transformation where both the input and the output are bipartite systems. Applying the ontological model as described in~\eqref{comp-atomic-mu} to such a transformation yields
\begin{align}
	&\input{BCT-diagrams/4_Bipar.tikz}=\input{BCT-diagrams/12_mud27.tikz}\nonumber\\
	=&\input{BCT-diagrams/12_mud28.tikz}\;\;,
\end{align}
which simplifies to
\begin{equation}
	\input{BCT-diagrams/4_Bipar.tikz}=\input{BCT-diagrams/4_Bipar8.tikz}\;\;,
\end{equation}
which is in agreement with the multi-system definition in Eq.~\eqref{ont-model-on-composite-atomics}.

\subsubsection{Diagram preservation of $\mu$}
\label{diagram-preservation-of-mu}
For a pure bipartite state followed by $\nu$, diagram preservation implies that
\begin{equation}
	\input{BCT-diagrams/12_mud30.tikz}=\input{BCT-diagrams/12_mud13.tikz}=\input{BCT-diagrams/12_mud14.tikz}\;.
\end{equation}
This can be justified in this specific case, because applying the ontological model—regardless of whether diagram preservation is assumed, as in the leftmost diagram—must yield the same result as the rightmost diagram. The latter can be shown to coincide with the case where diagram preservation is assumed, that is, by first applying the ontological model to the bipartite pure state given by~\eqref{ont-model-on-composite-state}, and then applying $\mu$, as shown in the middle diagram. This confirms that diagram preservation holds in this instance.

For multipartite pure states, diagram preservation implies
\begin{align}
	\input{BCT-diagrams/12_mud32.tikz}&=\input{BCT-diagrams/12_mud33.tikz}\nonumber\\[0.4cm]
	&=\input{BCT-diagrams/12_mud31.tikz}\;.
\end{align}
The first diagram corresponds to the use of the ontological model, regardless of whether or not the model satisfies diagram preservation. This is equal to the final diagram—i.e., when the ontological model is used on the pure state resulting from the action of $\nu$. On the other hand, according to our definition of the ontological model on a pure composite states in~\eqref{ont-model-on-composite-state}, one can show that in the case where diagram preservation holds—as in the second diagram in the equation above—if one first applies the ontological model to the original composite pure state, and then applies $\mu$, the result will match that of the last diagram. Therefore, diagram preservation holds in this specific instance, and by recursively applying this reasoning, one can derive~\eqref{composite-satate-as-single-states}. The same argument applies to a composite pure effect in equation~\eqref{composite-effects-as-single-effects}, where using what we define in~\eqref{ont-model-on-composite-effect} one can show that diagram diagram preservation holds when one has
\begin{align}
	\input{BCT-diagrams/12_mud34.tikz}&=\input{BCT-diagrams/12_mud35.tikz}\nonumber\\[0.4cm]
	&=\input{BCT-diagrams/12_mud36.tikz}\;.
\end{align}
Using the above equation recursively, one is able to derive~\eqref{composite-effects-as-single-effects}.

Using the same logic, for a normalised atomic transformation on composite systems, diagram preservation implies
\begin{align}
	\input{BCT-diagrams/12_mud38.tikz}&=\input{BCT-diagrams/12_mud39.tikz}\nonumber\\[0.4cm]
	&=\input{BCT-diagrams/12_mud40.tikz}\;.
\end{align}
Which using~\eqref{ont-model-on-composite-atomics} ensures that diagram preservation indeed holds. Using the equation above recursively, one can obtain~\eqref{comp-atomic-mu}.

\vspace{3mm}

\section{Atomic transformations on composite systems}
\label{atomiccomp}
An atomic transformation on composite systems~\eqref{atomic-comp} is atomic, if and only if, for every ancillary system $E$ with dimension $n_E$. one has
\begin{align}
	\label{atomic-on-composite-systems}
	&\input{BCT-diagrams/13_BCT35.tikz}\nonumber\\
	=&\lambda\delta_{i'_1,i_1}\delta_{i'_2,i_2}\cdots\delta_{i'_p,i_p}\delta_{s'_1,s_1}\delta_{s'_2,s_2}\delta_{s'_{p-1},s_{p-1}}\input{BCT-diagrams/13_BCT36.tikz}\;,
\end{align}
Where $1\leq j_E \leq n_E$, and $s_E\in\lbrace0,1\rbrace$. For example, for an atomic transformation with a bipartite system input and a bipartite system output, one has
\begin{equation}
	\label{bipartite-atomic}
	\input{BCT-diagrams/13_BCT13.tikz}=\lambda\delta_{i,i_1}\delta_{j,i_2}\delta_{s,s_1}\input{BCT-diagrams/13_BCT14.tikz}\; .
\end{equation}
To show why a transformation $\mathscr{A}$ on composite systems is atomic if and only if~\eqref{atomic-on-composite-systems} holds, recall~\eqref{nu-on-multipartite-atomic} (originally stated for normalised atomic transformations, but it can be straightforwardly extended to atomic transformations by multiplying both sides by a scalar factor $\lambda$, which does not affect the argument). In this setting, $\mathscr{A}$ is expressed in terms of an atomic transformation $\mathscr{A}_s$ on the corresponding single system, and for clarity we denote this as $\mathscr{A} = \nu_p \circ \mathscr{A}_s \circ \nu^{-1}_{q}$, where $\nu_p$ denotes the consecutive application of $\nu$ for $p$ times, and $\nu^{-1}_{q}$ denotes the consecutive application of $\nu^{-1}$ for $q$ times.

To prove that $\mathscr{A}$ is atomic if and only if it satisfies~\eqref{atomic-on-composite-systems}, is equivalent of proving, $\mathscr{A}$ is atomic if and only if $\mathscr{A}_s$ is atomic. To see this equivalence, consider a transformation $\mathscr{A}$ on a composite system that satisfies~\eqref{atomic-on-composite-systems}, without assuming atomicity. Then, by~\eqref{atomic-on-composite-systems}, there exists a single-system transformation $\mathscr{A}_s$ such that $\mathscr{A}$ can be expressed in terms of $\mathscr{A}_s$ as $\mathscr{A} = \nu_p \circ \mathscr{A}_s \circ \nu^{-1}_{q}$. Since $\mathscr{A}_s$ satisfies~\eqref{atomic-trans}, it is atomic. Therefore, to conclude that $\mathscr{A}$ is atomic, it suffices to show that the atomicity of $\mathscr{A}_s$ implies the atomicity of $\mathscr{A}$. Also, to show the other direction, that an atomic transformation $\mathscr{A}$ satisfies~\eqref{atomic-on-composite-systems}, we can define a transformation $\mathscr{A}_s$ such that one can rewrite $\mathscr{A}$ in terms of $\mathscr{A}_s$ as $\mathscr{A} = \nu_p \circ \mathscr{A}_s \circ \nu^{-1}_{q}$. If we can show that the atomicity of $\mathscr{A}$ implies the atomicity of $\mathscr{A}_s$, then it follows that $\mathscr{A}_s$ satisfies~\eqref{atomic-trans}, which ultimately implies that $\mathscr{A}$ satisfies~\eqref{atomic-on-composite-systems}. Hence, to prove that $\mathscr{A}$ is atomic if and only if it satisfies~\eqref{atomic-on-composite-systems}, one must prove $\mathscr{A}$ is atomic if and only if $\mathscr{A}_s$ is atomic.

This indeed follows from the properties of $\nu$. We can use proof by contradiction to show that the atomicity of $\mathscr{A}_s$ implies the atomicity of $\mathscr{A}$. For that, suppose $\mathscr{A}_s$ is atomic, but $\mathscr{A}$ is not. This means there exist two transformations $\mathscr{A}_1$ and $\mathscr{A}_2$ such that $\mathscr{A} = \mathscr{A}_1 + \mathscr{A}_2$ and $\mathscr{A}_1 \not\propto \mathscr{A}_2$. But this leads to contradiction, since one has $\mathscr{A}_s = \nu^{-1}_q \circ \mathscr{A}_1 \circ \nu_p + \nu^{-1}_q \circ \mathscr{A}_2 \circ \nu_p$, and using the properties of $\nu$ implies that $\mathscr{A}_s$ is not atomic, because $\nu^{-1}_q \circ \mathscr{A}_1 \circ \nu_p \not\propto \nu^{-1}_q \circ \mathscr{A}_2 \circ \nu_p$. Hence, $\mathscr{A}$ is atomic. The same logic can be applied to show the other direction. That is, the atomicity of $\mathscr{A}$ implies the atomicity of $\mathscr{A}_s$. If one assumes the contrary, that is $\mathscr{A}_s$ is not atomic, then one would be led to a contradiction of having $\mathscr{A}$ as a non atomic transformation. Therefore, $\mathscr{A}$ is atomic if and only if $\mathscr{A}_s$ is atomic. Moreover, one can show that every transformation $t$ on composite systems can be uniquely expressed as a conical combination of atomic transformations on composite systems. This follows directly from Proposition~\ref{uniqueness}, since any transformation on composite systems can be mapped to a transformation on single systems via $\nu$. Consequently, one obtains a unique conical combination of atomic transformations $\mathscr{A}_s$ on single systems, which in turn implies that the conical combination of the corresponding composite normalised atomic transformations $\mathscr{A}$ is also unique for the original transformation.

\section{Consistency checks proofs}\label{App:Consistency}
\subsection{Diagram preservation proof}
\label{dpproof}
In App.~\ref{diagram-preservation-of-mu}, we have shown that the ontological model is diagram-preserving when applying $\mu$. Here, we aim to generalise this property and demonstrate that diagram preservation holds in every arbitrary scenario, not just in the specific cases we previously examined.

Since the ontological model is linear and every transformation in BCT can be expressed as a conical combination of atomic transformations, verifying that the ontological model preserves diagrams under sequential composition reduces to checking this property specifically for the sequential composition of atomic transformations. Recall that the sequential composition of atomic transformations \eqref{seq-comp-atomic} is given by
\begin{equation}
	\input{BCT-diagrams/13_BCT23.tikz}=\input{BCT-diagrams/13_BCT24.tikz}\; ,
\end{equation}
where $\tilde{\tau} =\tau\oplus \tau'$ and $\tilde{\lambda}=\delta_{l,i'_{0}}\lambda\lambda '$. If diagram preservation holds for the sequential composition of atomic transformations, using linearity, one must have
\begin{equation}
	\input{BCT-diagrams/14_dp.tikz}=\delta_{l,i'_{0}}\input{BCT-diagrams/14_dp1.tikz}\;,
\end{equation}
for all systems $n$, $q$, and $m$, all 
$1 \le i_0 \le n$, $1 \le l \le q$, $1 \le i'_0 \le q$, $1 \le l' \le m$, 
and all $\tau,\tau' \in \{0,1\}$. Using~\eqref{ont-model-on-single-atomics}, this can be shown to hold, as one has
\begin{equation}
	\input{BCT-diagrams/14_dp2.tikz}=\delta_{l,i'_{0}}\input{BCT-diagrams/14_dp3.tikz}\;.
\end{equation}
Similarly, using $\mu$ in the form it appears in~\eqref{comp-atomic-mu}, along with the above result, ensures that diagram preservation holds for the sequential composition of atomic transformations when systems are composite. Hence, diagram preservation also holds for the sequential composition of arbitrary transformations when systems are composite.

For parallel composition, first we can rewrite parallel composition of two arbitrary transformations as
\begin{equation}
	\label{dp1}
	\input{BCT-diagrams/9_GenDiag.tikz}=\input{BCT-diagrams/9_GenDiag1.tikz}=\input{BCT-diagrams/9_GenDiag2.tikz}\;,
\end{equation}
where we used diagram preservation for sequential composition. If one has
\begin{equation}
	\label{dp2}
\begin{split}
	\input{BCT-diagrams/9_GenDiag2.tikz}&=\input{BCT-diagrams/9_GenDiag3.tikz}=
    \\[1ex]&=
    \input{BCT-diagrams/9_GenDiag4.tikz}=\input{BCT-diagrams/9_GenDiag5.tikz}\;.
\end{split}
\end{equation}
Then diagram preservation holds for parallel composition. Hence, to verify the above equality, we must check diagram preservation for an arbitrary bipartite system in which one subsystem undergoes a transformation and the other remains unchanged, while also verifying preservation of the swap between subsystems and of the identity on each subsystem. 

Since the transformation in question can be rewritten as a conical combination of normalised atomic transformations, and since the ontological model is linear, it suffices to check diagram preservation for normalised atomic transformations.
\begin{equation}
	\label{diagram-preservation-atomic}
	\input{BCT-diagrams/4_OntMapping.tikz}=\input{BCT-diagrams/4_OntMapping1.tikz}=\input{BCT-diagrams/4_OntMapping6.tikz}\;.
\end{equation}
For every system $n_1$, $n_2$, $m$, and for every $1\leq i_0 \leq n_1$,  $1\leq l \leq n_2$, $\sigma\in\lbrace 0,1\rbrace$.
To verify this, one can write
\begin{equation}
	\label{atranslinebelow}
	\input{BCT-diagrams/1_AtTranWithLineBelow.tikz}=\sum_{i_{1},i_{2},s_{1},l_{1},l_{2},t_{1},\tau}C_{i_0}^{l,\sigma}\left( i_{1},i_{2},s_{1},l_{1},l_{2},t_{1},\tau\right)\input{BCT-diagrams/14_dp4.tikz}\;,
\end{equation}
which is the conical combination of normalised atomic transformations for the transformation on the left-hand side of the equation above. Applying the left-hand side to an arbitrary pure tripartite state with the ancillary system $E$ gives
\begin{equation}
	\label{LHS}
\begin{split}
	\input{BCT-diagrams/3_AtTranOn3ParS.tikz}&=\input{BCT-diagrams/3_AtTranOn3ParS1.tikz}=
    \\[1ex]&=
    \delta_{i,i_{0}}\input{BCT-diagrams/3_AtTranOn3ParS2.tikz}=\delta_{i,i_{0}}\input{BCT-diagrams/3_AtTranOn3ParS3.tikz}\;.
\end{split}
\end{equation}
For the first and the last equality, we have used the relabeling of tripartite pure states $\left(\left(i j \right)_{s}k\right)_{t}=\left(i\left(jk\right)_{s \oplus t}\right)_{s}$, as defined for BCT in~\cite[Eq.~38]{d2020classicality}, where it is shown that all relevant consistency checks for this relabeling are satisfied in BCT. Applying the right-hand side of~\eqref{atranslinebelow} to the same pure tripartite state gives
\begin{align}
	\label{RHS}
	&\sum_{i_{1},i_{2},s_{1},l_{1},l_{2},t_{1},\tau}C_{i_0}^{l,\sigma}\left( i_{1},i_{2},s_{1},l_{1},l_{2},t_{1},\tau\right)\input{BCT-diagrams/3_AtTranOn3ParS4.tikz}\nonumber\\[1ex]
	=&\sum_{i_{1},i_{2},s_{1},l_{1},l_{2},t_{1},\tau}C_{i_0}^{l,\sigma}\left( i_{1},i_{2},s_{1},l_{1},l_{2},t_{1},\tau\right)\delta_{i,i_{1}}\delta_{j,i_{2}}\delta_{s,s_{1}}\input{BCT-diagrams/3_AtTranOn3ParS5.tikz}\;,
\end{align}
which by putting \eqref{LHS} and \eqref{RHS} equal, we have
\begin{equation}
	C_{i_0}^{l,\sigma}\left( i_{1},i_{2},s_{1},l_{1},l_{2},t_{1},\tau\right)=\delta_{\sigma,\tau}\delta_{i_2,l_2}\delta_{l_1,l}\delta_{i_{1},i_{0}}\delta_{t_1,s_{1}\oplus\tau}\;.
\end{equation}
Substituting the above coefficients in (\ref{atranslinebelow}) gives us
\begin{equation}
	\label{AtLineBelow}
	\input{BCT-diagrams/1_AtTranWithLineBelow.tikz}=\sum_{l_2,s_{1}}\input{BCT-diagrams/2_BiPartAtTrans1.tikz}\; .
\end{equation}
Now, to verify \eqref{diagram-preservation-atomic}, we must also examine identity preservation. To see this, first note that in BCT, identity, for every system $m$, is given by
\begin{equation}
	\begin{tikzpicture}
	\begin{pgfonlayer}{nodelayer}
		\node [style=none] (0) at (-1.5, 0) {};
		\node [style=none] (1) at (1.5, 0) {};
		\node [style=none] (2) at (0, 0.5) {};
		\node [style=none] (3) at (0, 0.5) {\scriptsize$m$};
	\end{pgfonlayer}
	\begin{pgfonlayer}{edgelayer}
		\draw [bWire] (0.center) to (1.center);
	\end{pgfonlayer}
\end{tikzpicture}
=\sum_{i}\input{BCT-diagrams/8_PresOfId.tikz} \;.
\end{equation}
This can be verified by applying the both sides to an arbitrary pure bipartite state with the subsystem $m$ and an ancillary system. By applying the ontological model to both sides of the equation above, one has
\begin{equation}
	\input{BCT-diagrams/8_PresOfId6.tikz}=\sum_{i}\input{BCT-diagrams/8_PresOfId2.tikz}=\sum_{i}\input{BCT-diagrams/14_dp5.tikz}=\input{BCT-diagrams/14_dp6.tikz}=\begin{tikzpicture}
	\begin{pgfonlayer}{nodelayer}
		\node [style=none] (0) at (-1.5, 0) {};
		\node [style=none] (1) at (1.5, 0) {};
		\node [style=none] (2) at (0, 0.5) {\scriptsize$\tilde{\xi}(m)$};
	\end{pgfonlayer}
	\begin{pgfonlayer}{edgelayer}
		\draw [oWire] (1.center) to (0.center);
	\end{pgfonlayer}
\end{tikzpicture}
\;.
\end{equation}
Hence, the ontological model is identity preserving. Having identity preservation, if we apply the ontological model to the left-hand side of~\eqref{atranslinebelow}, then, assuming diagram preservation holds, one must have
\begin{equation}
	\label{DiagPresAtLine}
	\input{BCT-diagrams/4_OntMapping.tikz}=\input{BCT-diagrams/4_OntMapping1.tikz}=\input{BCT-diagrams/4_OntMapping6.tikz}=\input{BCT-diagrams/14_dp10.tikz}\;.
\end{equation}
On the other hand, by applying the ontological model to the right-hand side of \eqref{AtLineBelow}, regardless of whether the ontological models is diagram preserving, one has
\begin{equation}
	\sum_{l_2,s_1}\input{BCT-diagrams/4_OntMapping3.tikz}=\sum_{l_2,s_1}\input{BCT-diagrams/14_dp8.tikz}=\input{BCT-diagrams/14_dp10.tikz}\;,
\end{equation}
where the last equality can be verified when we apply both sides to the orthonormal basis of the corresponding classical space. Hence, it coincides with~\eqref{DiagPresAtLine}, which implies that the ontological model is indeed digram preserving for~\eqref{DiagPresAtLine}, which ultimately implies diagram preservation for an arbitrary bipartite system where one subsystem undergoes a transformation while the other remains unchanged, i.e.,
\begin{equation}
	\label{transformation-and-identity}
	\input{BCT-diagrams/14_dp16.tikz}=\input{BCT-diagrams/14_dp17.tikz}\;.
\end{equation}

The final step that must be checked for parallel composition, is the swap preservation. This is analogous to what we have done previously, which is to write the swap in terms of conical combinations of bipartite normalised atomic transformations, for every system $n$ and $m$, as
\begin{equation}
	\label{SwapAtom}
	\input{BCT-diagrams/6_SwapBCT.tikz}=\sum_{i_{1},i_{2},s_{1},l_{1},l_{2},t_{1},\tau}C\left( i_{1},i_{2},s_{1},l_{1},l_{2},t_{1},\tau\right)\input{BCT-diagrams/14_dp12.tikz}\;.
\end{equation}
Recalling~\eqref{swap}, by applying the left-hand side of the above equation to an arbitrary pure tripartite state with an ancillary system $E$, one obtains
\begin{equation}
	\input{BCT-diagrams/6_SwapOn3ParS.tikz}=\input{BCT-diagrams/6_SwapOn3ParS1.tikz}\;,
\end{equation}
and by applying the right-hand side, one has
\begin{align}
	&\sum_{i_{1},i_{2},s_{1},l_{1},l_{2},t_{1},\tau}C\left( i_{1},i_{2},s_{1},l_{1},l_{2},t_{1},\tau\right)\input{BCT-diagrams/14_dp14.tikz}\nonumber\\[1ex]
	=&\sum_{i_{1},i_{2},s_{1},l_{1},l_{2},t_{1},\tau}C\left( i_{1},i_{2},s_{1},l_{1},l_{2},t_{1},\tau\right)\delta_{i,i_1}\delta_{j,i_2}\delta_{s,s_1}\input{BCT-diagrams/14_dp13.tikz}\;,
\end{align}
which gives
\begin{equation}
	C\left( i_{1},i_{2},s_{1},l_{1},l_{2},t_{1},\tau\right)=\delta_{i_{2},l_2}\delta_{l_2,i_1}\delta_{t_1,s_1}\delta_{\tau,s_1} .
\end{equation}
Substituting these conical coefficients in~\eqref{SwapAtom} results in 
\begin{equation}
	\input{BCT-diagrams/6_SwapBCT.tikz}=\sum_{l_1,l_2,t_1}\input{BCT-diagrams/6_SwapOn3ParS3.tikz}\;.
\end{equation}
Now we are in a position to see if the ontological model preserves the swap, and for this one needs to apply the ontological model to the both sides of equation above, which gives
\begin{equation}
\begin{split}
	\input{BCT-diagrams/7_OntMapOfSwap.tikz}&=\sum_{l_1,l_2,t_1}\input{BCT-diagrams/7_OntMapOfSwap1.tikz}=
    \\[1ex]&=\sum_{l_1,l_2,t_1}\input{BCT-diagrams/14_dp15.tikz}=\input{BCT-diagrams/7_OntMapOfSwap6.tikz}\;.
\end{split}
\end{equation}

Hence, the ontological model is swapping preserving. Swapping preservation, together with~\eqref{transformation-and-identity}, implies that~\eqref{dp2} holds. Therefore, the ontological model is diagram preserving for the parallel composition of arbitrary transformations. Similarly, using $\mu$ the way it has been used in~\eqref{comp-atomic-mu}, one can show that diagram preservation also holds for the parallel composition of arbitrary transformations when systems are composite.

Note that here we have only shown diagram preservation in the case where every system, whether a single system or a composite system, is non-trivial. Nevertheless, since our ontological model defines how pure states, pure effects, normalised atomic transformations, and scalar $1$ are mapped, one can use the same reasoning to show that diagram preservation also holds when at least one system is trivial. In other words, in such cases, it is possible to breakdown an arbitrary composition into a conical combination of normalised atomic transformations, pure effects, pure states, and scalar $1$. Then, by the linearity of the ontological model, one can verify diagram preservation for arbitrary compositions involving at least one trivial system.

\subsection{Determinacy preservation proof}
\label{detproof}
Unique deterministic effect in BCT for a system $n$ is given by
\begin{equation}
	\begin{tikzpicture}
	\begin{pgfonlayer}{nodelayer}
		\node [style=infupground] (0) at (0.25, 0) {};
		\node [style=none] (1) at (-1, 0) {};
		\node [style=none] (2) at (-0.5, 0.25) {\scriptsize$n$};
	\end{pgfonlayer}
	\begin{pgfonlayer}{edgelayer}
		\draw [bWire] (1.center) to (0);
	\end{pgfonlayer}
\end{tikzpicture}
=\sum_{i=1}^{n}\input{BCT-diagrams/13_BCT15.tikz}\;,
\end{equation}
Using the definition of the ontological model for effects in \eqref{effects}, one can verify that the ontological model preserves the unique deterministic effect on each system $n$, since
\begin{equation}\label{eq:deterministic-effect-preservation}
	\input{BCT-diagrams/2_Discarding.tikz}=\sum_{i}\input{BCT-diagrams/1_1ParSt5.tikz}=\sum_{i}\input{BCT-diagrams/10_effects3.tikz}=\input{BCT-diagrams/10_effects8.tikz}=\begin{tikzpicture}
	\begin{pgfonlayer}{nodelayer}
		\node [style=infupground] (0) at (0.5, 0) {};
		\node [style=none] (1) at (-1.25, 0) {};
		\node [style=none] (2) at (-0.5, 0.5) {\scriptsize$\tilde{\xi}(n)$};
	\end{pgfonlayer}
	\begin{pgfonlayer}{edgelayer}
		\draw [oWire, in=360, out=180] (0) to (1.center);
	\end{pgfonlayer}
\end{tikzpicture}
\;,
\end{equation}
where the penultimate equality follows from the fact that the sum over all pure effects of system $n$ in the ontological model gives the unique deterministic effect for system $n$. Since the ontological model is digram preserving, for an arbitrary deterministic transformation $t$ from system $n$ to $m$~\eqref{eq:normalisation}, i.e.,
\begin{equation}
	\input{BCT-diagrams/10_effects4.tikz}=\begin{tikzpicture}
	\begin{pgfonlayer}{nodelayer}
		\node [style=infupground] (0) at (2, 0) {};
		\node [style=none] (1) at (0.75, 0) {};
		\node [style=none] (2) at (1.25, 0.5) {\scriptsize$n$};
	\end{pgfonlayer}
	\begin{pgfonlayer}{edgelayer}
		\draw [bWire] (0) to (1.center);
	\end{pgfonlayer}
\end{tikzpicture}
\;,
\end{equation}
Applying the ontological model to both sides, and using diagram preservation along with the uniqueness of the deterministic effect, implies that the corresponding transformation in the ontological model is also deterministic, i.e.,
\begin{equation}
	\input{BCT-diagrams/10_effects6.tikz}=\begin{tikzpicture}
	\begin{pgfonlayer}{nodelayer}
		\node [style=infupground] (0) at (0.5, 0) {};
		\node [style=none] (1) at (-1.25, 0) {};
		\node [style=none] (2) at (-0.5, 0.5) {\scriptsize$\tilde{\xi}(n)$};
	\end{pgfonlayer}
	\begin{pgfonlayer}{edgelayer}
		\draw [oWire, in=360, out=180] (0) to (1.center);
	\end{pgfonlayer}
\end{tikzpicture}
\;,
\end{equation}
which, using $\mu$~\eqref{eq:mu}, one can show that it also holds true for transformations involving composite systems. Hence, the ontological model of BCT is determinacy-preserving. This ensures that the ontological model maps every arbitrary instrument in BCT to a valid instrument in classical theory as well, which follows from coarse-graining preservation. Since the full coarse-graining of all transformations within a given instrument yields a deterministic transformation, linearity of the ontological model and determinacy preservation imply that the coarse-graining of the corresponding transformations in classical theory is also deterministic. Therefore, the corresponding instrument in classical theory is itself a valid instrument.

\section{Reversible transformations}
\label{App:rev-trans}
\subsection{In bilocal classical theory}
\label{App:rev-trans-bct}
For all systems $n\neq 1$, a transformation $\mathscr{R}$ is reversible in BCT (see Postulate~3 and Proposition~1 in~\cite{d2020classicality}) if and only if, there exists a permutation $\pi$ of $n$ elements and a binary variable $\sigma_i \in \lbrace 0,1 \rbrace$, such that for all non-trivial ancillary systems $E$, and all bipartite pure states $\vert (ij)_{s} )$ of the system of interest and ancilla, one has
\begin{equation}
	\label{reversible-transformations}
	\input{BCT-diagrams/16_rev11.tikz}=\input{BCT-diagrams/1_RevTransBCT1.tikz}\;.
\end{equation} 
Hence, we can denote a reversible transformation on system $n$ as
\begin{equation}
	\input{BCT-diagrams/1_RevTrans.tikz}\;.
\end{equation}
\subsection{In the ontological model}
\label{App:rev-trans-ont}
Here we derive the ontological representation of reversible transformations in BCT. This is not strictly necessary for any of our results, but shows example calculations that can be performed using the ontological representation, and provides a tool which will be useful for future research on applications of BCT. 

Writing an arbitrary reversible transformation on system $n$ as a conical combination of normalised atomic transformations, one has
\begin{equation}
	\label{single-rev}
	\input{BCT-diagrams/1_RevTrans.tikz}=\sum_{i_{0},l,\tau}C^{\sigma}_{\pi}(i_{0},l,\tau)\input{BCT-diagrams/15_ea2.tikz}\;.
\end{equation}
After applying the both sides of the equation above to an arbitrary bipartite pure state of system $n$ and ancilla $E$, one has
\begin{equation}
	\input{BCT-diagrams/1_RevTransBCT.tikz}=\sum_{i_{0},l,\tau}C^{\sigma}_{\pi}(i_{0},l,\tau)\input{BCT-diagrams/2_RevTransAtomic1.tikz}=\sum_{i_{0},l,\tau}C^{\sigma}_{\pi}(i_{0},l,\tau)\delta_{i,i_{0}}\input{BCT-diagrams/2_RevTransAtomic2.tikz}\;,
\end{equation}
which, when put equal to the right-hand side of equation~\eqref{reversible-transformations}, one has
\begin{equation}
	C^{\sigma}_{\pi}(i_{0},l,\tau)=\delta_{\pi\left(i_{0}\right),l}\delta_{\sigma_{i_{0}},\tau}.
\end{equation}
Substituting the above coefficients in~\eqref{single-rev} and applying the ontological model, one has
\begin{equation}
	\label{RevSing}
	\input{BCT-diagrams/5_RevElemOnt.tikz}=\sum_{i_0}\input{BCT-diagrams/5_RevElemOnt1.tikz}=\sum_{i_0}\input{BCT-diagrams/16_rev.tikz}=\input{BCT-diagrams/16_rev1.tikz}\;,
\end{equation}
where
\begin{equation}
	\input{BCT-diagrams/16_rev2.tikz}:=\input{BCT-diagrams/16_rev3.tikz}\;,\; \input{BCT-diagrams/16_rev4.tikz}:=\input{BCT-diagrams/16_rev5.tikz}\;.
\end{equation}
The subscript $S$ denotes that $\sigma_S$ is acting upon the single system $n$. The last equality in~\eqref{RevSing} can be verified when one applies both sides to the orthonormal basis of the corresponding classical space. 

Using $\mu$~\eqref{mu-derivation}, we can also derive how the ontological model maps reversible transformations on composite systems. For example, for a bipartite reversible transformation $\mathscr{R}$, one has
\begin{equation}
	\input{BCT-diagrams/5_RevElemOnt4.tikz}=\input{BCT-diagrams/16_rev6.tikz}\;,
\end{equation}
where $\mathscr{R}s$ denotes the corresponding reversible transformation on a single system. Before examining how the ontological model maps a bipartite reversible transformation, we first need to specify how to denote a permutation of a pure bipartite state in BCT. To this end, consider applying the bipartite reversible transformation to a pure tripartite state, with an arbitrary ancillary system $E$:
\begin{align}
	\input{BCT-diagrams/16_rev13.tikz}
    &=\input{BCT-diagrams/16_rev14.tikz}
    =
    \nonumber\\[1ex]&=\input{BCT-diagrams/16_rev15.tikz}
	=
    \nonumber\\[1ex]&=\input{BCT-diagrams/16_rev16.tikz}\;,
\end{align}
where $\sigma_{Q(i,j,s)}\in \lbrace 0,1\rbrace$, and $\pi\left(Q\left(i,j,s\right)\right)$ denotes a permutation of $2nm$ elements. Note that we have assumed that if a reversible transformation $\mathscr{R}$ on bipartite systems can be expressed in terms of a transformation $\mathscr{R}_s$ on the corresponding single system as $\mathscr{R}=\nu\circ\mathscr{R}_s\circ\nu^{-1}$, then $\mathscr{R}_s$ is reversible itself, where we used the reversibility of $\mathscr{R}_s$ in the last equality of the equation above. This indeed holds, as the existence of the inverse of $\mathscr{R}$ implies that its corresponding single-system transformation is the inverse of $\mathscr{R}_s$. Therefore, $\mathscr{R}_s$ is indeed reversible.

Since $\pi\left(Q\left(i,j,s\right)\right)$ is also representing a permutation of a bipartite pure state, we define
\begin{equation}
	\input{BCT-diagrams/16_rev17.tikz}:= \input{BCT-diagrams/16_rev18.tikz}\;.
\end{equation}
In other words, the permuted bipartite pure state corresponding to the permuted single pure state $\pi(Q(i,j,s))$ can be denoted as $\vert (i'j')_{s'} )$, where one has $1 \leq i' \leq n$, $1 \leq j' \leq m$, and $s' \in \{0,1\}$, each being a function of the original indices of the bipartite state, which we denote this functional relation as $i' = \pi_n(i,j,s)$, $j' = \pi_m(i,j,s)$, and $s' = \pi_B(i,j,s)$.  Therefore, one has 
\begin{equation}
	\input{BCT-diagrams/16_rev13.tikz}=    \input{BCT-diagrams/16_rev19.tikz}\;,
\end{equation} 
Where $\sigma_{i,j,s}=\sigma_{Q(i,j,s)}$. Hence, we can denote a reversible transformation on a bipartite system $n\boxtimes m$ as $\mathscr{R}^{\sigma}_{\pi_n,\pi_m,\pi_B}$, and its corresponding reversible transformation on a single system as $\mathscr{R}s_{\pi_Q}^{\sigma_Q}$. Now, if we apply the ontological model, we have
\begin{align}
	\label{bipartite-reversible-transformation}
	\input{BCT-diagrams/16_rev20.tikz}&=\input{BCT-diagrams/16_rev21.tikz}\nonumber\\
	=& \input{BCT-diagrams/16_rev22.tikz}\nonumber\\
	\vspace{0.6pt}\nonumber\\
	=&\input{BCT-diagrams/16_rev23.tikz}\;,
\end{align}
wherein one has
\begin{equation}
	\scalebox{0.8}{\input{BCT-diagrams/16_rev25.tikz}}=\scalebox{0.8}{\input{BCT-diagrams/16_rev26.tikz}}\;,
    \scalebox{0.8}{\input{BCT-diagrams/16_rev27.tikz}}=\scalebox{0.8}{\input{BCT-diagrams/16_rev28.tikz}}\;,
    \scalebox{0.8}{\input{BCT-diagrams/16_rev29.tikz}}=\scalebox{0.8}{\input{BCT-diagrams/16_rev30.tikz}}\;,
\end{equation}
and
\begin{equation}
	\input{BCT-diagrams/16_rev9.tikz}:=\input{BCT-diagrams/16_rev12.tikz}\;.
\end{equation}
Alternatively,  we can express any reversible transformation on bipartite systems as a conical combination of normalised atomic transformations on bipartite systems. By determining the explicit form of the conical coefficients and applying the ontological model, one derives the same result as in~\eqref{bipartite-reversible-transformation}.

\end{document}